\documentclass[12pt]{article}

\usepackage{amsfonts}
\usepackage{amssymb}
\usepackage{amsmath}

\usepackage{amsthm}

\usepackage{authblk}
\usepackage[titletoc]{appendix}

\addcontentsline{toc}{section}{Appendices}

\newcommand{\bm}[1]{\mbox{\boldmath $#1$}}

\def\gNH{g_{\mbox{\tiny NHG}}}

\def\f{F}
\def\lie{{\pounds}}

\def\G{{\cal G}}
\def\n{\bar{D}}

\def\S{\Sigma}
\def\k{\kappa_\xi}

\def\be{\begin{equation}}
\def\ee{\end{equation}}
\def\bea{\begin{eqnarray}}
\def\eea{\end{eqnarray}}
\def\bean{\begin{eqnarray*}}
\def\eean{\end{eqnarray*}}

\def\={\stackrel{\Sigma}{=}}

\newlength{\cellwidth}
\usepackage{tabularx}
\usepackage{multirow}
\usepackage{array}

\newcolumntype{P}[1]{>{\centering\arraybackslash}p{#1}}
\newcolumntype{M}[1]{>{\centering\arraybackslash}m{#1}}

\def\eqHeta{\stackrel{\H_{\eta}}{=}}

\def\Htwo{\H_{\eta}}

\def\Kill{{\mathcal A}}
\def\eqH{\stackrel{\H_{\xi}}{=}}

\def\H{{\mathcal H}}

\newtheorem{theorem}{Theorem}
\newtheorem{definition}{Definition}
\newtheorem{remark}{Remark}
\newtheorem{corollary}{Corollary}
\newtheorem{proposition}{Proposition}
\newtheorem{lemma}{Lemma}

\newcounter{mnotecount}[section]

\renewcommand{\themnotecount}{\thesection.\arabic{mnotecount}}
\newcommand{\mnote}[1]
{\protect{\stepcounter{mnotecount}}$^{\mbox{\footnotesize
$
\bullet$\themnotecount}}$ \marginpar{
\raggedright\tiny\em
$\!\!\!\!\!\!\,\bullet$\themnotecount: #1} }

\newcommand{\mnotex}[1]
{\protect{\stepcounter{mnotecount}}$^{\mbox{\footnotesize $\bullet$\themnotecount}}$
\marginpar{
\raggedright\tiny\em
$\!\!\!\!\!\!\,\bullet$\themnotecount: #1} }

\begin{document}

\title{Multiple Killing Horizons and Near Horizon Geometries}
\author[1]{Marc Mars}
\author[2]{Tim-Torben Paetz}
\author[3]{Jos\'e M. M. Senovilla}
\affil[1]{Instituto de F\'isica Fundamental y Matem\'aticas, Universidad de Salamanca, Plaza de la Merced s/n, 37008 Salamanca, Spain}
\affil[2]{Gravitational Physics, University of Vienna, Boltzmanngasse 5, 1090 Vienna, Austria}
\affil[3]{Departamento de F\'isica Te\'orica e Historia de la Ciencia, Universidad del Pa\'is Vasco UPV/EHU, Apartado 644, 48080 Bilbao, Spain}

\maketitle

\vspace{-0.2em}

\begin{abstract}
  Near Horizon Geometries with multiply degenerate Killing horizons $\H$ are considered, and their degenerate Killing vector fields identified. We prove that they all arise from hypersurface-orthogonal Killing vectors of any cut of $\H$ with the inherited metric ---cuts are spacelike co-dimension two submanifolds
  contained 
in $\H$. For each of these Killing vectors on a given cut, there are three different
possibilities for the Near Horizon metric which are presented explicitly. The structure of the metric for Near Horizon Geometries with multiple Killing horizons of order $m\geq 3$ is thereby completely determined, and in particular we prove that the cuts on $\H$ must be warped products with maximally symmetric fibers (ergo of constant curvature).
The question whether multiple degenerate Killing horizons may lead to inequivalent Near Horizon Geometries by using different degenerate Killings is addressed, and answered on the negative: all Near Horizon geometries built from a given multiple degenerate Killing horizon (using different degenerate Killings) are isometric.
\end{abstract}





\section{Introduction}

In a recent paper \cite{MPS} we have introduced the notion of multiple
Killing horizon (MKH) and have initiated a systematic study of its properties. In essence,
multiple Killing horizons are null hypersurfaces which are
simultaneously Killing horizons of two or more Killing vectors. The precise
definition, recalled in the next section,
is slightly more involved as one needs to take care of the fact
that different generators can have different fixed points.
Several general properties
of multiple Killing horizons were obtained in \cite{MPS}. In particular, one can attach a natural number $m \geq 2$ to each  MKH, called order, which counts the number of its linearly independent Killing generators.  The order
of any MKH cannot be larger than the dimension of the spacetime where it
lies, and examples exist for MKHs of any allowed order.
Another important property of MKHs is that the surface gravity of each  generator is constant, and  that at most one
(in an appropriate sense) can  be different from zero.
MKHs are called fully degenerate or non-fully degenerate depending
on whether all its surface gravities vanish or not.
The  order of a fully degenerate MKH can be any number
strictly smaller than the spacetime dimension,  and again examples of
any  order exist
\cite{MPS}.

An example of paramount importance of MKH is the Killing horizon of a near
horizon geometry (NHG) spacetime. These spacetimes are obtained
by infinite zoom of the geometry around
any Killing horizon with vanishing surface gravity.
Thus, NHG describe the ``focused'' geometry of degenerate Killing horizons.
They turn out to be very interesting objects both from a geometric
and from a physical perspective  and have been extensively studied
(see \cite{KL} and references therein). One of its general properties is the so-called ``enhancement of symmetry'':
in addition to the Killing vector associated to the original
degenerate Killing horizon (which is preserved by the ``zoom'' limit),
any NHG always admits a second Killing vector \cite{KL,PLJ,LSW} with respect
to which the horizon is non-degenerate. This immediately turns all horizons
associated to NHG spacetimes into multiple Killing horizons.

Now, what happens
if the original horizon (before taking the near horizon limit) is a
multiple Killing horizon itself?  We have proved in \cite{MPS} that any
of its degenerate Killing generators survives to the limit.
Hence, if the original multiple Killing horizon was fully-degenerate of order $m$, it follows that the near horizon limit has a non-fully degenerate
Killing horizon of order $m+1$. Correspondingly, if the original horizon was non-fully degenerate
of order  at least three, then the near horizon limit has at least the same
order.

This fact raises the following natural question, posed in \cite{MPS}.
Since the original Killing horizon is degenerate with respect to more than
one linearly independent Killing vector and any one of them  can be used to
perform the near horizon limit, are all the NHG one obtains by this process
(locally) isometric to each other or not?
If the answer were no,  i.e. if the near horizon geometry
depended on the choice of degenerate Killing generator, one could
iterate the process and generate a potentially very large class
of near horizon geometries starting from a single MKH. If, on the other hand,
the answer is yes (the limit is independent of the choice of generator) one
concludes that any degenerate Killing horizon (multiple or not)
has a well-defined and unique near horizon geometry attached to it. One of
the main objectives of this paper is to answer this question and prove that
 any degenerate Killing horizon (multiple or not) defines a unique
near horizon geometry.

The strategy we follow is to find the explicit
coordinate change (i.e. local isometry) that transforms one near horizon geometry limit into another. Despite its apparent simplicity, the problem turns out to be substantially more involved than one could have expected.
The key to success is the ability to obtain
very explicit  and fully general information on
near horizon geometries admitting
 multiple Killing horizons of order at least three.
Finding these results is our
second main achievement   and opens up the possibility of,
eventually, finding a  complete classification of all
near horizon geometries admitting multiple Killing horizons with at least two
degenerate generators. To be more specific, recall that near horizon geometries of dimension $n+1$
are determined by a Riemannian manifold $(S,\gamma)$ of dimension $n-1$
endowed with a one-form $\bm{s}$ and a scalar $h$.  We show that
whenever the horizon of the near horizon geometry is multiple
of order $m$ at least three, the geometry of $S$ is (locally)
a warped product with fibers of dimension $m-2$ of constant sectional
curvature. In addition, there are $m-2$ linearly independent hypersurface-orthogonal
Killing vectors of $(S,\gamma)$ tangent to the fibers
which, together with the trivial zero vector field,
are in one-to-one correspondence to the degenerate Killing generators
of the MKH. The one-form $\bm{s}$ and scalar $h$ are completely and
explicitly determined in terms
of the geometric properties of any one of the non-trivial hypersurface-orthogonal
Killing vectors of $(S,\gamma)$.

The plan of the paper is as follows. In section \ref{basics} we summarize the
main results of \cite{MPS} that are needed in this work.
In  section \ref{sec:NHGKillings} we recall the limit process that leads to the near
horizon geometry  and discuss in which sense it is
determined by the data $(S,\gamma,\bm{s},h)$.
We also recall the key result in \cite{MPS} showing that the near horizon limit does not
reduce the order of multiple Killing horizons. We then find the
equations that need to be satisfied by any degenerate
Killing generator of the horizon. They  involve the near horizon
data $(S,\gamma,\bm{s},h)$ and the proportionality function $f$ on $S$
between generators.
The full set of equations include
the so-called master equation
 that must necessarily
be satisfied
by all MKH \cite{MPS}. In the near horizon geometry case this
is supplemented by two more equations
involving $h$. We
find the general solution of the full system and show that the solutions
are related to hypersurface orthogonal Killing vectors in $(S,\gamma)$ (theorem
\ref{th:sexact}
and lemma \ref{ODESol}). In section \ref{Sect:structure}
we exploit these results to prove that
$(S,\gamma)$ is a warped product with fibers of constant curvature
(theorem \ref{th:warped}) which
in particular means that if the  MKH is of order $m \geq 3$ then $(S,\gamma)$
admits at least $(m-1)(m-2)/2$ linearly independent Killing vectors.
The value of the constant curvature of those fibers is explicitly determined (theorem \ref{th:signofK}).
Finally, we devote section \ref{Uniqueness} to proving that given
any
near horizon geometry spacetime with a multiple Killing horizon of order
at least three, any choice of degenerate generator leads, via
the standard near horizon limit, to a spacetime which
is  locally isometric to the original one (theorem
\ref{th:isom}). As already mentioned,
we show this by finding explicitly the coordinate change that transforms
one metric into another. We first find necessary geometric
conditions that must  be
satisfied by the coordinate change, and which ultimately determines it
in an essentially unique way. We then prove
that this coordinate change  indeed defines 
a (local) isometry between the two spacetimes. This step requires exploiting
the explicit information about near horizon geometries with MKH
obtained in previous sections. We finish the paper
by stating and proving our main theorem, namely that to any multiple Killing horizon one can attach a unique near horizon geometry (theorem
\ref{th:main}).

\subsection{Notation}
$(M,g)$ is a spacetime, that is, a connected, oriented and time-oriented $(n+1)$-dimensional Lorentzian manifold with metric $g$ of signature $(-,+,\dots,+)$. All submanifolds will be without boundary and the topological closure of a set $A$ is denoted by $\overline{A}$. Given a vector (field) $v$ in $TM$, $\bm{v}$ denotes the metrically related one-form.
Moreover, $\mathfrak{X}(N)$ denotes the set of smooth vector fields on a differentiable manifold $N$.
We use index-free and index notation.
Lowercase Greek letters $\alpha,\beta,\dots $ are spacetime indices and run from $0$ to $n$.
Capital Latin indices $A,B,\dots$ are co-dimension-2 submanifold indices running from $2$ to $n$. Small Latin indices $i,j,\dots$ will enumerate 
either (1) the different Killing vectors of multiple Killing horizons then taking values in $\{1,\dots , m\}$, where $m\leq n+1$, or (2) the hypersurface-orthogonal Killing vectors in $(S,\gamma)$ generating the maximally symmetric fibers, in which case they take values in $\{1,\dots ,p\}$ with $p\leq n-1$. In index-free notation covectors and vectors metrically related to each other will be
distinguished using boldface for covectors.

\section{Basics on Multiple Killing Horizons}
\label{basics}
We start by recalling the notions of Killing Horizon and Bifurcate Killing Horizon of a spacetime $(M,g)$
of dimension at least two.
\begin{definition}[Killing horizon of a Killing $\xi$]
A {\bf Killing horizon of a Killing vector $\xi$} of $(M,g)$
is a smooth embedded null hypersurface $\H_{\xi}$
such that $\xi|_{\H_\xi}$ is null, nowhere zero and tangent to $\H_\xi$.
In case $\H_\xi$ has more than one connected component we also demand that the interior of its closure is a smooth connected hypersurface.
\end{definition}
The reason to allow for multiple connected components was discussed in \cite{MPS}.
Any Killing horizon $\H_\xi$ possesses the notion of {\bf surface gravity} $\kappa_{\xi}$ defined by
\begin{align}
\nabla_{\xi} \xi \eqH  \kappa_{\xi} \xi \quad
\quad \mbox{or equivalently} \quad \quad \mbox{grad} (g(\xi,\xi)) \eqH  -2 \kappa_{\xi} \xi
\label{SurGrav}.
\end{align}
If $\k$ vanishes, then $\H_\xi$ is called {\em degenerate}.

\begin{definition}[Bifurcate Killing horizon]
Let $S$ be a connected co-dimension two spacelike submanifold of fixed points of a Killing vector $\xi$. The set of points along all null geodesics orthogonal to $S$ comprises a {\bf bifurcate Killing horizon} \cite{Boyer,KW,RW} with respect to $\xi$.
\end{definition}
A bifurcate Killing horizon is composed of five pieces: two connected Killing horizons $\H_1^+$ and $\H_2^+$ to the future of $S$ ---not including $S$---, two more connected Killing horizons $\H_1^-$ and $\H_2^-$ to the past of $S$, and $S$ itself.
Notice that $\H_1^+ \cup \H_1^- $ is a Killing horizon according to the previous definition, and its closure (adding $S$) is a  connected null  hypersurface; and similarly for $\H_2^+ \cup \H_2^-$.

In \cite{MPS} we introduced the following class of Killing Horizons.
\begin{definition}[Multiple Killing horizon (MKH)]\label{def:MKH}
A {\bf multiple Killing horizon} in $(M,g)$ is an embedded null hypersurface $\H$ such that there exist
Killing horizons  $\H_{\xi_i}$ associated to linearly independent Killing vectors $\xi_i$, $i\in\{1,\dots ,m\}$ with $m\geq 2$, with the property
\begin{align*}
\overline{\H} = \overline{\H}_{\xi_1} = \dots = \overline{\H}_{\xi_m}.
\end{align*}
\end{definition}

As proven in \cite{MPS}, the set of all Killing vectors in $(M,g)$ with a common multiple Killing horizon constitute a Lie sub-algebra, denoted by $\Kill_\H$, of the Killing Lie algebra. Its dimension $m:=$dim $\Kill_\H \geq 2$ is called the order of the MKH.
%
Sometimes we write `double, triple,' etcetera, MKHs for $m=2,3$ etc.

The following fundamental theorem was proven in \cite{MPS}.
\begin{theorem}\label{th:algebra}
$\Kill_\H$ contains an Abelian sub-algebra $\Kill_\H^{deg}$ of dimension at least $m-1$ whose elements all have vanishing surface gravities. If $\Kill_\H^{deg} \neq \Kill_\H$, the remaining independent Killing vector (say $\xi$) in $\Kill_\H\setminus \Kill_\H^{deg}$ has a constant surface gravity $\k\neq 0$ and satisfies
\be
\left[\xi,\eta \right] =-\k \eta , \hspace{1cm} \forall \eta \in \Kill_\H^{deg}.
\ee
\end{theorem}


Thus, there are two essentially inequivalent possibilities: {\bf fully degenerate} MKH if $\Kill_\H = \Kill_\H^{deg}$, in which case its Lie algebra is Abelian and all surface gravities vanish; and {\bf non-fully degenerate} MKHs, with an essentially unique non-zero surface gravity.

In general, the maximum possible dimension of $\Kill_\H^{deg}$ is $n=$ {\em dim}$(M)-1$. Therefore, the maximum possible order of a MKH $\H$ is $m=n$ for fully degenerate $\H$,
and $m=n+1$ for non-fully degenerate $\H$.

\section{Degenerate Killing Vectors on Near Horizon Geometries}\label{sec:NHGKillings}
The near horizon geometry of a degenerate Killing horizon $\H_{\eta}$ is usually defined as follows \cite{KL}: nearby $\H_{\eta}$ with degenerate Killing vector $\eta$, local Gaussian null coordinates $\{v,u,x^A\}$ can be chosen such that the metric reads
$$
g = 2dv\left(du + 2 u\,  \check s_A dx^A +\frac{1}{2} u^2 \check h dv\right) +\check \gamma_{AB} dx^A dx^B
$$
where $\check h$, $\check s_A$, and $\check\gamma_{AB}$ are independent of $v$, the degenerate Killing  reads Scan$\eta=\partial_v$ and the degenerate Killing horizon has been placed at $\H_{\eta} = \{u=0\}$.
Replacing $v\rightarrow v/\lambda$ and $u\rightarrow u\lambda$ here and taking the limit $\lambda \rightarrow 0$ one is led to the metric of its ``near-horizon'' geometry
\be
g_{\mbox{\tiny NHG}}
= 2dv\left(du + 2 u\,  s_A dx^A +\frac{1}{2} u^2 h dv\right) +\gamma_{AB} dx^A dx^B \label{NHG}
\ee
where $h=\check h|_{u=0}$, $s_A=\check s_A|_{u=0}$ and $\gamma_{AB}=\check\gamma_{AB}|_{u=0}$.
This is the ``focused'' local geometry near $\H_{\eta}$.

\begin{remark}\label{intrinsic}
The NHG of a degenerate Killing horizon $\H_{\eta}$ can be intrinsically and geometrically defined as follows: Pick up any co-dimension two submanifold $S\subset \H_{\eta}$ (we call these {\em cuts}). Then
\begin{itemize}
\item $\gamma$ is the first fundamental form on $S$
\item $\bm{s}$ is the {\em torsion one-form} on $S$, defined by $\bm{s}(V) := \bm{\ell} (\nabla_V \eta)$ for any $V\in \mathfrak{X}(S)$, where $\ell$ is uniquely determined by the conditions $g(\ell,V)=0 \, \, \, \forall V\in \mathfrak{X}(S)$, $g(\ell,\ell)=0$ and $g(\ell,\eta)=-1$.
\item $h=2 \gamma^{\sharp}(\bm{s},\bm{s}) - \mbox{{\rm div}} \, s + \frac{1}{2} R|_{S}-
\frac{1}{2}{\rm tr}_\gamma Ric |_{S}$
\end{itemize}
where $\gamma^{\sharp}$ is the contravariant metric associated to $\gamma$,
{\rm div} is the divergence on $(S,\gamma)$, $R$ is the scalar curvature and $Ric$ the Ricci tensor of $(M,g)$, both pull-backed to $S$.
\end{remark}
One can check
 that the scalar curvature $R$ of the metric $g$ coincides with the scalar curvature of the metric $g_{\mbox{\tiny{NHG}}}$ at every cut $S$, and similarly for the term ${\rm tr}_\gamma Ric |_{S}$.

 The construction of a near horizon geometry relies on Gaussian null coordinates associated to the degenerate Killing vector $\eta$.
 These coordinates cannot cover domains where  $\eta$ has fixed points.
 Since we assume all spacetimes (in particular the NHG spacetime)
 to be connected, the Gaussian null coordinates leading to the NHG
 can cover at most
 one connected component of $\H_{\eta}$ and this is a proper subset
 of $\overline{\H}_{\eta}$ whenever the latter has fixed points of $\eta$.
 When the limit is performed to compute the NHG only this portion
 of $\H_{\eta}$ is considered whence the  degenerate Killing vector $\eta=\partial_v$
 of the NHG never has fixed points. The NHG does not see the rest
 of the original Killing horizon. The intrinsic definition
 in Remark~\ref{intrinsic}  has exacly the same limitation
 because the normalization condition $g(\ell,\eta)=-1$ cannot be fulfilled
 in  domains where $\eta$ has zeros. In this paper we want to understand the NHG limit of multiple Killing horizons, so we must restrict to domains
 of the horizon which are connected and contain no fixed points of any
 of the degenerate Killing generators under consideration. Our results are valid only on those domains. To understand the global picture one would need to devise
 a way of defining NHG that allows for fixed points of the degenerate Killing.
 This is an interesting problem that deserves consideration but it is beyond the scope of this paper. 

Any near-horizon geometry in the above sense possesses a non-fully degenerate MKH $\H_{\mbox{\tiny NHG}}$ because
\begin{enumerate}
\item  The original degenerate Killing $\eta$ leads, after the limit, to a Killing vector  which is also degenerate, and, by definition of the NHG, with $\H_{\mbox{\tiny NHG}} =\H_{\eta}$. 
  This follows easily from the explicit local expression of the metric (\ref{NHG}) as the Killing vector $\eta=\partial_v$ satisfies $g_{\mbox{\tiny NHG}} (\eta,\eta)=u^2 h$ so that $\eta$ is null on $\{u=0\}$
 and its surface gravity vanishes: $\kappa_{\eta}:= \mbox{grad} (g(\eta,\eta))\eqHeta 0$.
\item The metric (\ref{NHG}) always has another Killing vector given by \cite{KL,PLJ,LSW}
$$
\xi = v\partial_v -u\partial_u
$$
which is null on, and tangent to, $\H_{\eta}=\{u=0\}$ except at its set of fixed points $S_\xi\supset \{u=v=0\}$. Thus, $\H_{\xi} = \H_{\eta}\setminus S_\xi$ is a Killing horizon for $\xi$ with several connected components but such that $\overline{\H}_{\xi}=\H_{\eta}=\H_{\mbox{\tiny NHG}}$, 
and therefore $\H_{\mbox{\tiny NHG}}$ is a MKH of order $m\geq 2$.
\end{enumerate}
The commutator is
$$
[\xi,\eta] = -\eta
$$
hence Theorem \ref{th:algebra} implies that $\H_{\mbox{\tiny NHG}}$ is non-fully degenerate with $\k =1$. Actually, any cut $S_{v_0}:=\{u=0,v=v_0\}$
of $\H_{\eta}$ is the bifurcation surface of a bifurcate Killing horizon with bifurcation Killing vector $\xi -v_0\eta $.

In \cite{MPS} we established the following theorem.
\begin{theorem}
\label{NearHor}
Let $\H$ be a multiple Killing horizon of order $m$ and
$(M_{\mbox{\tiny NHG}}, g_{\mbox{\tiny{NHG}}})$
be the near-horizon
geometry of a degenerate Killing vector $\eta\in \Kill_\H^{deg}$ . Then
\begin{itemize}
\item[(i)] If $\H$ is fully degenerate,
$(M_{\mbox{\tiny NHG}}, g_{\mbox{\tiny{NHG}}})$ admits a multiple Killing horizon $\H_{\mbox{\tiny NHG}}$
of order at least $m+1$.
\item[(ii)] If $\H$ is non-fully degenerate and $m \geq 3$, then
$(M_{\mbox{\tiny NHG}}, g_{\mbox{\tiny{NHG}}})$  has a multiple Killing horizon $\H_{\mbox{\tiny NHG}}$
of order at least $m$.
\end{itemize}
\end{theorem}
Of course, the theorem also holds if $\H$ is non-fully degenerate and of order $m=2$, but then the result is trivial, as the MKH $\H_{\mbox{\tiny NHG}}$ of all NHGs have $m=2$ at least.

As briefly discussed in \cite{MPS}, this theorem raises the natural question of whether or not the NHG spacetime $(M_{\mbox{\tiny NHG}}, g_{\mbox{\tiny{NHG}}})$ arising from
a multiple Killing horizon $\H$ is independent
of the choice of  $\eta\in \Kill_\H^{deg}$. To address this problem, in the next subsection we identify the NHGs that possess a MKH of order $m\geq 3$ as well as their corresponding degenerate Killing vectors.

We have already mentioned above that the NHG taken w.r.t.\ a degenerate Killing vector $\eta^{(1)}$ only takes (in general) a proper subset of
the MKH ${\H}$ into account.
A second   degenerate Killing vector $\eta^{(2)}$ may have a different fixed point set on  ${\H}$
so that the NHGs computed from $\eta^{(1)}$ and $\eta^{(2)}$ work
in general in different subsets of the MKH of the original spacetime.
For this reason when analyzing (local) isometry of the NHGs we will consider
connected portions of the  MKH where both degenerate Killings have no fixed points.

\subsection{Degenerate Killing vectors of MKH in NHGs}
We start with a metric of type (\ref{NHG}), which holds
around a connected component of $\H_\eta$, and derive the equations for the existence of degenerate Killing vector fields other than $\eta =\partial_v$ there.
\begin{proposition}\label{prop:KillingEqs}
Any Killing vector $\zeta$ of the metric (\ref{NHG}) which has (the appropriate dense subset of) $ \H_\eta=\{u=0\}$ as degenerate Killing horizon must take the form
\be
\zeta = f\partial_v +\frac{u^2}{2} \Delta f \partial_u - u\,  {\rm grad} f \label{degKill}
\ee
where $\Delta$ and {\rm grad} are the Laplacian and gradient on any cut $S_0\subset \H_\eta$, and the function $f$ satisfies the following relations:
\bea
D_A D_B f &=&s_A D_B f +s_B D_A f , \label{master}\\
D_A f D^A h &=& 2 h s^AD_A f, \label{eq1}\\
h D_A f &=& 2D^B f  \left(D_B s_A -D_A s_B \right)+D_A (s^B D_B f) - 2 s_A s^B D_B f,  \label{eq2}
\eea
where $D_A$ is the covariant derivative on $(S_0,\gamma)$.

Conversely, for any function $f$ which solves (\ref{master})-(\ref{eq2}) the vector field  (\ref{degKill}) belongs
to $\Kill_{\H_\eta}^{deg}$.
\end{proposition}
\begin{remark}
Equation (\ref{master}) was found in full generality (for arbitrary MKHs) in \cite{MPS} and called the master equation.
\end{remark}
\begin{remark}
The  degenerate  Killing vector $\zeta$ given in (\ref{degKill}) has fixed points on $\H_{\eta}$ if and only if the function $f$ has zeros.
As described above the NHG computed from $\zeta$ is only defined where $\zeta$ has no fixed points, whence we will be mainly interested in the subset $\H_{\eta,\zeta}:=\{p\in \H_{\eta} : f(p)\ne 0\}$.
This will be relevant in Section~\ref{Uniqueness}.
\end{remark}
\begin{proof}
Let $\zeta \in \Kill_{\Htwo}^{deg}$, i.e.  any degenerate Killing generator
of $\Htwo$ in the metric (\ref{NHG}), and set
$$\zeta = f\partial_v + q \partial_u + \zeta^B \partial_B$$
there.
Then we know that
\be
q|_{u=0} =0, \hspace{1cm} \zeta^B|_{u=0}=0. \label{atu=0}
\ee
It follows from theorem \ref{th:algebra}  that $[\eta,\zeta]=0$ so that
$$
\partial_v f = \partial _v q = \partial_v \zeta^B =0.
$$
Consider the Killing equations
$$
(\lie_\zeta g)_{\mu\nu} = \zeta^\rho \partial_\rho g_{\mu\nu} +g_{\mu\rho}\partial_\nu \zeta^\rho +g_{\rho\nu}\partial_\mu \zeta^\rho =0.
$$
Letting $\mu =u$ (with some abuse of notation), these relations for $\nu =u,v,A$ become, respectively,
\begin{align}
& \partial_u f =0 \quad \quad  \Longrightarrow \quad \quad f=f(x^A) , \label{fxA}\\
& \partial_u q +2us_A \partial_u \zeta^A =0 , \label{qu}\\
& \partial_A f +\gamma_{AB} \partial_u \zeta^B =0 .\label{Df}
\end{align}
Similarly, letting $\mu =v$, the equations for $\nu =v,A$ become respectively
\begin{align}
& 2uqh+u^2 \zeta^B \partial_B h =0, \label{preeq1}\\
& 2qs_A +2u\zeta^B\partial_B s_A +\partial_A q +u^2 h \partial_A f +2 u s_B \partial_A \zeta^B =0.\label{preeq2}
\end{align}
Finally, for $\mu =A$ and $\nu =B$ we get
\be
\zeta^C \partial_C \gamma_{AB} +\gamma_{AC}\partial_B \zeta^C +\gamma_{CB}\partial_A \zeta^C +2u s_A \partial_B f +2u s_B \partial_A f =0. \label{premaster}
\ee
Given (\ref{fxA}) the function $f$ can be seen as a function on the cut $S_0$, and then $D_A f =\partial_A f$. Taking this into account together with (\ref{atu=0}), the solution of (\ref{Df}) reads
$$
\zeta^A =-u \gamma^{AB} D_B f = -u (\mbox{grad} f)^A
$$
and then (\ref{qu}) with (\ref{atu=0}) provides
\be
q = u^2 s^A D_A f
\ee
while (\ref{preeq1}) becomes (\ref{eq1}) and (\ref{preeq2}) becomes (\ref{eq2}). The remaining equation (\ref{premaster}) can now be written as
$$
-(\pounds_{grad \, f} \gamma)_{AB}+2 s_A D_B f+ 2 s_B D_A f =0
$$
which leads directly to (\ref{master}). To finish the proof is enough to note that the trace of (\ref{master}) gives
\be
2s^B D_B f = \Delta f  \label{Deltaf}.
\ee

For the converse one simply checks that all components (\ref{fxA})-(\ref{premaster}) of the Killing equation are satisfied assuming that
(\ref{master})-(\ref{eq2}) hold.
\end{proof}
Observe that the solution $f=$ const.\ provides the original Killing $\eta=\partial_v$. We can now derive the main result in this subsection.
\begin{theorem}\label{th:sexact}
Let $(M_{\mbox{{\tiny NHG}}},g_{\mbox{{\tiny NHG}}})$
be a near horizon geometry with metric (\ref{NHG}). Then, the vector field (\ref{degKill}), where $f$ is a smooth non-constant scalar in $S_0$, is a degenerate Killing generator of the chosen connected component $ \Htwo$ if and only if
the following two conditions hold:
\begin{itemize}
\item[(i)]  The differential of $f$ vanishes nowhere  and
the metric $\gamma$ on a cut $S_0\subset \{u=0\}= \Htwo$ (and therefore on any such cut)
admits a hypersurface-orthogonal Killing vector field
\be
\bm{\varsigma} =\frac{Q(f)}{N} df \label{killing},
\ee
where $N := \gamma^{\sharp}(df,df)$ is the square norm of $df$ and $Q(f)$
is a not identically zero solution of the system of ODEs
\begin{align}
\frac{dQ}{df} = Q(f) P(f), \quad \quad \frac{d}{d f} \left ( \frac{d P}{df} + P^2 \right ) + P \left ( \frac{dP}{df} + P^2 \right ) =0.
\label{ODEs}
\end{align}
 \item[(ii)] The torsion one-form $\bm{s}$ and metric coefficient $h$ take the form
   \begin{align}
\bm{s} = \frac{1}{2} \left ( \frac{dN}{N} - P df \right ), \quad \quad
h = \frac{1}{2} N \left ( \frac{dP}{df} + P^2 \right ).
\label{sh}
\end{align}
In particular, $\bm{s}$ is closed.
\end{itemize}
\end{theorem}
\begin{proof}
Assume that the vector field (\ref{degKill})  is a degenerate Killing vector of $\{u=0\}$ in
$(M_{\mbox{{\tiny NHG}}},g_{\mbox{{\tiny NHG}}})$.
From Proposition \ref{prop:KillingEqs},   equations (\ref{master}--\ref{eq2}) are satisfied. We first observe that (\ref{master}) implies that
if $D_A f$ vanishes at a point then it vanishes everywhere. Since by assumption $f$ is non-constant we conclude that $df\neq 0$ everywhere on $S_0$.
Contracting (\ref{master}) with $D^Bf$ one then derives
%
%
\be
s_A= \frac{1}{2} D_A(\ln N) -\frac{1}{2} \frac{D_Af}{N} \Delta f \label{s1}
\ee
where we have used (\ref{Deltaf}), and the definition of $N$.
 Contracting here with $D^Af$ gives
\be
D^Af D_AN = 2N \Delta f  \label{Deltaf1}.
\ee
A similar contraction of (\ref{eq2}) using (\ref{Deltaf}) provides
\be
h=\frac{1}{2} \frac{D^Af D_A(\Delta f)-(\Delta f)^2}{N}.\label{h1}
\ee
Equation (\ref{s1}) can be rewritten as
\be
\bm{s} = \frac{1}{2} d \ln N - \frac{1}{2} \frac{\Delta f}{N} df \label{s2}
\ee
from where we derive
\be
d\bm{s} =-\frac{1}{2} d\left(\frac{\Delta f}{N} \right) \wedge df \label{ds}.
\ee
Contraction with grad$f$ here leads to
$$
2D^B f  \left(D_B s_A -D_A s_B \right)=-D^Bf D_B \left(\frac{\Delta f}{N} \right)D_A f + N  D_A \left(\frac{\Delta f}{N} \right)
$$
and introducing this into eq.(\ref{eq2}), making use of (\ref{h1}), (\ref{s1}) and
(\ref{Deltaf1}), and after a little calculation, we get
\be
\frac{1}{N^2} \left[D^Bf D_B(\Delta f) - 2 (\Delta f)^2\right] D_A f =  D_A \left(\frac{\Delta f}{N} \right) .\label{24}
\ee
This informs us that
$$
\Delta f/N :=P(f)
$$
is a function of $f$ so that (\ref{s2}) takes the form given in (\ref{sh}) and, from (\ref{ds}) follows that
$\bm{s}$ is closed. Combining (\ref{h1}) with (\ref{24}) and this notation we arrive at
\begin{equation}
h =\frac{N}{2} \left(\frac{dP}{df} +P^2 \right) \label{h2bis}
\end{equation}
which is the second in (\ref{sh}). The remaining equation is (\ref{eq1}), which on using (\ref{h2bis}) and (\ref{Deltaf1}) leads, after another computation, to the second in (\ref{ODEs})
$$
\frac{d}{df} \left(\frac{dP}{df} +P^2 \right)+P \left(\frac{dP}{df}+P^2\right) =0.
$$
Note that, introducing the function $Q(f)$ as defined in the first of (\ref{ODEs}),
\begin{align}
\bm{s} = \frac{1}{2} \frac{Q}{N} d \left ( \frac{N}{Q} \right ) \label{exps}
\end{align}
so that the master equation becomes, after rearranging and dividing by $N/Q$
\begin{equation}
D_A\left( \frac{Q}{N} D_B f \right)+ D_B\left( \frac{Q}{N} D_A f\right)=0.
\label{Killequation}
\end{equation}
This states that (\ref{killing}) defines a hypersurface
orthogonal Killing vector of $(S_0,\gamma)$ and the only if
part of the proof is completed.

For the converse, we assume that (i) and (ii) hold.
The master equation (\ref{master}) is satisfied due to
(\ref{Killequation}). Concerning (\ref{eq1}) and (\ref{eq2}), we first
observe that
$$
\varsigma(f) = \frac{Q}{N} \gamma^\sharp (df,df) = Q$$
hence
\begin{align}
\varsigma(N) =
2 \varsigma^B (D_A D_B f ) D^A f =
2 D^A f D_A ( \varsigma^B D_B f) = 2 D^A f D_A Q(f) = 2 Q P N
\label{inter}
\end{align}
where in the second equality we used that $\varsigma$ is a Killing vector.
An equivalent way to state (\ref{inter}) is
\begin{align}
D^A f D_A N = 2 P N^2.
\label{dfdN}
\end{align}
Similarly, one has
\begin{align}
s^A D_A f = \frac{1}{2} \left ( \frac{D_A N}{N} - P(f) D_A f \right ) D^A f = \frac{1}{2} P N
\label{sDf}
\end{align}
so that, using that $\bm{s}$ is closed, the right-hand
 side of
(\ref{eq2}) becomes
\begin{align*}
D_A (s^B D_B f) - 2 s_A s^B D_B f & =
D_A \left ( \frac{1}{2} P N \right ) - \left ( \frac{D_A N}{N} - P D_A f \right)
\frac{1}{2} P N \\
& =
\frac{1}{2} N \left ( \frac{dP}{df} + P^2 \right ) D_A f
\end{align*}
and (\ref{eq2}) holds because  $h$ is given by  (\ref{sh}).
Finally, we check (\ref{eq1}):
\begin{align*}
D^A f D_A h & - 2 h s^A D_A f =
\frac{1}{2} D^A f D_A N \left ( \frac{dP}{df} + P^2 \right )
+ \frac{1}{2} N D^A f \frac{d}{df} \left (
\frac{dP}{df} + P^2 \right ) D_A f \\
& -  N \left ( \frac{dP}{df} + P^2 \right ) \frac{1}{2} NP
=
\frac{1}{2} N^2 \left [
\frac{d}{df} \left ( \frac{dP}{df} + P^2 \right ) + P
\left ( \frac{dP}{df} + P^2 \right ) \right ] =0
\end{align*}
where in the second equality we used (\ref{dfdN}) and in the last one
(\ref{ODEs}).
In summary, equations (\ref{master})-(\ref{eq2}) hold and, by Proposition
\ref{prop:KillingEqs}, the vector field (\ref{degKill}) is a degenerate Killing generator of $\{u=0\}\subset \Htwo$. This proves the converse, because
the function $f$ is by assumption
non-constant (in fact with nowhere zero gradient).
\end{proof}

\begin{remark}
\label{remarka}
By Theorem~\ref{th:sexact} $\bm{s}$ is closed.
Note, though, that under the conditions of this theorem, equations
(\ref{exps}), (\ref{inter}), (\ref{dfdN}), (\ref{sDf})  and
\begin{align*}
\Delta f = 2 s^A D_A f = P N
\end{align*}
hold true.  Given that $N$ vanishes nowhere, this equation implies in particular
that $P(f)$ is smooth everywhere on $S_0$.
Therefore, the first equation  in (\ref{sh}) combined with the first in (\ref{ODEs}) states that $\bm{s}$ is exact on $S_0$.

It seems important to emphasize that to establish exactness of $\bm{s}$  it is crucial that $\eta$ does, by definition of the NHG, not have  fixed points on   $\H_{\mbox{{\tiny NHG}}}$.
In \cite{MPS2} we will consider a similar setting, where, though, fixed points are possible, and in that case one can only deduce that $\bm{s}$ is exact on cuts of
$\H_{\eta}$ as proper subsets of cuts of  $\H_{\mbox{{\tiny NHG}}}$.
The reason for that is that the gauge (Gaussian null coordinates) becomes  singular at the fixed points of $\eta$.
This will be analyzed in more detail in \cite{MPS2}.
 \end{remark}

In the next lemma, we find the most general solution of the ODE system
(\ref{ODEs}).
\begin{lemma}
\label{ODESol}
 $Q(f)$ and $P(f)$ solve the system (\ref{ODEs}) if and only if,  with $Q_0 \in \mathbb{R}\setminus\{0\}$,
they belong to  one of the following three exclusive cases:
\begin{itemize}
\item[(a)] $P=0$, $Q= Q_0$, and then
$\bm{s} = dN/(2N),\, \,  h =0$.
\item[(b)] $P = 1/(f+c)$, $Q = Q_0 (f+c)$ with
$c \in \mathbb{R}$, and then
\begin{align*}
\bm{s} = \frac{1}{2} \left (
\frac{dN}{N} - \frac{df}{f+c}  \right ), \quad \quad h =0.
\end{align*}
\item[(c)] $\displaystyle{P = \frac{2 (f+c)}{b + (f+c)^2}}$ and $Q = Q_0 [b + (f+c)^2]$
with $b,c \in \mathbb{R}$. Then
\begin{align}
\bm{s} = \frac{1}{2} \left ( \frac{dN}{N}
- \frac{2 (f+c)}{b + (f+c)^2} df \right ), \quad \quad
h = \frac{N}{b + (f+c)^2}.
\label{expshc}
\end{align}
\end{itemize}
\end{lemma}

\begin{proof}
The expressions for $\bm{s}$ and $h$ in each case follow directly from (\ref{sh}) by simple substitution.
First of all,
$P=0$ is clearly a solution of (\ref{ODEs}). In such case $Q$ is a
non-zero constant and we fall into case $(a)$.
Assume then that $P$ is not identically zero and define $W(f)$ by
\begin{align*}
\frac{dP}{df} + P^2 := W P.
\end{align*}
The second equation in (\ref{ODEs}) is then
\begin{align*}
\frac{d}{df} \left ( W P  \right ) + P^2 W
= P \frac{d W}{df} + \frac{dP}{df} W + P^2 W =
P \left ( \frac{dW}{df} + W^2 \right ) =0.
\end{align*}
It is immediate that the general solution of this ODE is either $W=0$ or $W = 1/(f+c)$
where $c$ is a real constant. The first case corresponds to $(b)$ because
\begin{align*}
\frac{dP}{df} + P^2 =0  \quad ( \mbox{with} \quad P \neq 0)
\quad \quad
\Longleftrightarrow \quad \quad  P = \frac{1}{f+c}
\end{align*}
and   the integration of the first in (\ref{ODEs}) gives $Q = Q_0 ( f+c)$.

It remains the case $W = (f+c)^{-1}$,
 which will be (c). We need to solve
\begin{align*}
\frac{dP}{df} + P^2 = W P = \frac{P}{f+c}.
\end{align*}
This is a Ricatti  equation and its solution is easily found by  introducing here the first in (\ref{ODEs}) which yields
\begin{align}
\frac{d^2 Q}{df^2} = \frac{1}{(f+ c)} \frac{dQ}{df}
 \quad  \Longleftrightarrow \quad \frac{dQ}{df} = 2 Q_0 (f+c) \Longleftrightarrow
Q = Q_0 \left ( b + (f+c)^2 \right ), \quad b \in \mathbb{R}.
\label{Qcasec}
\end{align}
The condition $Q_0 \neq 0$ is required because
otherwise $\displaystyle{P= Q^{-1} \frac{dQ}{df}}$ would vanish identically and we would fall
into a previous case.  From the expression (\ref{Qcasec}) of $Q$ one immediately
finds
$ P= 2 (f+c)/(b + (f+c)^2)$ and that (\ref{expshc}) holds.
\end{proof}

\begin{remark}
  \label{remarkb}
As shown in Remark \ref{remarka}, $P (f)$ is smooth on $S_0$.
This implies that $f+c$ has a definite sign on $S_0$ in case $(b)$ and that
$b + (f+c)^2$ also has a definite sign in case $(c)$. Consequently $Q(f)$
vanishes nowhere and $\bm\varsigma$ has no fixed points. We can define a smooth
 positive function $M$ on $S_0$ by
 $$
 N := Q^2 M.
 $$
 This function  is invariant under the flow of $\varsigma$, that is, $\varsigma (M)=0$ ---equivalently it
satisfies $D^A f D_A M=0$--- as follows from
\begin{align*}
D^A f D_A M = D^A f D_A \left ( \frac{N}{Q^2} \right )
= \frac{2 P N^2}{Q^2} - \frac{2N}{Q^3} \frac{dQ}{df} D_A f D^A f =0.
\end{align*}
This informs us that $M$ is independent of $f$, and thus either $dM=0$ or $dM \wedge df \neq0$.
\end{remark}

\begin{definition}
Let $S_0$ be a cut of a connected component $\{u=0\}= \Htwo$ of the MKH of a NHG with local metric (\ref{NHG}). By $\Kill_{S_0}\subset \mathfrak{X}(S_0)$ we denote a collection of vector fields
$\varsigma\in \mathfrak{X}(S_0)$ which are Killing vectors of $(S_0,\gamma)$  and take the form
(\ref{killing}) with either $df=0$ (yielding the zero vector field, which we call trivial) or with $df$ nowhere zero,
$N = \gamma^{\sharp}(df,df)$ and $(Q(f),P(f))$ solving
(\ref{ODEs}) (ergo given by the explicit forms of lemma \ref{ODESol}), such that
(\ref{sh}) holds with fixed $\bm{s}$ and $h$.
\end{definition}

\begin{remark}
Note that $\Kill_{S_0}$ depends via (\ref{sh})  on $\bf{s}$ and $h$.  Different NHGs thus  may  select different Killing vectors $\varsigma$
which may even produce $\Kill_{S_0}$'s of  different dimension
(anticipating that  $\Kill_{S_0}$ is a vector space, cf.\ Proposition~\ref{VectorSpace} below).
\end{remark}

The following lemma shows that given any non-trivial  $\varsigma \in
\Kill_{S_0}$  the functions $f$ and $Q(f)$ are defined uniquely
up to a constant rescaling.

\begin{lemma}
\label{scaling}
Let $\varsigma \in \Kill_{S_0}$ be non-trivial
and let $f_a: S_0 \rightarrow
\mathbb{R}$ and $Q_a(f_a)$, $a=1,2$ be such that (\ref{killing}) and
(\ref{ODEs}) hold. Then there exist constants $\alpha,\beta$ with $\alpha
\neq 0$ such that $f_1 = \alpha f_2 + \beta$ and $Q_1 (f_1) = \alpha Q_2 (f_2)$. Furthermore, with obvious notations, $P_2(f_2)=\alpha P_1(f_1)$.
\end{lemma}
\begin{proof}
Let $Z_a := N_a/Q_a$ for $a\in \{1,2\}$. Each $Z_a$  is
nowhere zero and defined everywhere on $S_0$. From (\ref{exps}) it follows
\begin{align*}
\bm{s} = \frac{1}{2Z_1} dZ_1 = \frac{1}{2Z_2} dZ_2 \quad \quad
\Longrightarrow \quad \quad Z_1 = \alpha Z_2
\end{align*}
where $\alpha$ is a non-zero constant. Since $\displaystyle{\bm\varsigma = \frac{df_1}{Z_1}
= \frac{df_2}{Z_2}}$  it must be $df_1 = \alpha df_2$ and hence
$f_1 = \alpha f_2 + \beta$.  Expression $\alpha  Q_2(f_2) =
Q_1(f_1)$
 is now immediate because $N_1 = \alpha^2 N_2$. The first in (\ref{ODEs}) for each $Q_a$ then gives $Q_2P_2 =Q_1 P_1$ which provides $P_2(f_2)=\alpha P_1(f_1)$.  
\end{proof}
 Observe that the invariance of $h$ follows easily from the second in (\ref{sh}).

Combining this with lemma \ref{ODESol} the following corollary follows
by a simple computation.
\begin{corollary}
Letting $Q_0^{(a)}$, $c_a$ and $b_a$ be the constants defined by Lemma \ref{ODESol}
from the explicit form of $Q_a(f_a)$ in each of the cases (note that
the scaling transformation defined by $\alpha, \beta$ cannot change the case), the following scaling law is obtained
\begin{align*}
& \mbox{Case (a)} \quad \quad \quad Q_0^{(2)} = \frac{1}{\alpha} Q_0^{(1)}, \\
& \mbox{Case (b)} \quad \quad \quad Q_0^{(2)} = Q_0^{(1)},
\quad c_2 = \frac{c_1 + \beta}{\alpha}, \\
& \mbox{Case (c)} \quad \quad \quad Q_0^{(2)} = \alpha Q_0^{(1)},  \quad
c_2 = \frac{c_1 + \beta}{\alpha}, \quad b_2 = \frac{b_1}{\alpha^2}.
\end{align*}
\end{corollary}
As a consequence of these results, there exists a smooth non-zero function
$Z$  such that the {\em exact} torsion one-form reads
\be
\bm{s}=\frac{1}{2Z} dZ \label{s=dz}
\ee
and all non-trivial elements $\varsigma \in \Kill_{S_0}$ can be represented by functions $f$, $Q(f)$ satisfying
\begin{align*}
\frac{\gamma^{\sharp}(df,df)}{Q(f)} =Z.
\end{align*}
The function $Z$ is, itself, defined up to a constant rescaling but the
point is that once $Z$ is fixed once and for all, this choice is made independently of $f$. Therefore,
such a prescription freezes the scaling freedom defined by $\alpha$ in Lemma \ref{scaling}.
We shall make this choice from now on.

\begin{proposition} \label{VectorSpace} $\Kill_{S_0}$ is vector space
  contained in the set of hypersurface-orthogonal Killing vectors of $(S_0,\gamma)$ .
\end{proposition}
\begin{proof}
  That all elements of $\Kill_{S_0}$ are hypersurface orthogonal Killing vectors
  is obvious, so we need to prove that they form a vector space.
  Any two non-trivial elements
$\varsigma_1, \varsigma_2 \in \Kill_{S_0}$ can be written as
$\bm\varsigma_1 = Z^{-1} df_1$,
$\bm\varsigma_2 = Z^{-1} df_2$,  so that, for any $a_1, a_2 \in \mathbb{R}$,
$\bm\varsigma' := a_1 \bm\varsigma_1 + a_2\bm \varsigma_2 = Z^{-1} d (a_1 f_1 + a_2 f_2)$ is a hypersurface-orthogonal Killing vector field of $(S_0,\gamma)$.
We can assume that $f' := a_1 f_1 + a_2 f_2$ is non-constant (otherwise
$\varsigma'$ is the zero vector and belongs to $\Kill_{S_0}$ trivially).
To prove that $\varsigma'$ belongs to $\Kill_{S_0}$ observe that from the function $f'=a_1 f_1 + a_2 f_2$
we construct, using the expression (\ref{degKill}), the vector field
$$
\zeta'=f'\partial_v+\frac{u^2}{2} \Delta f' \partial_u -{\rm grad} f' = a_1 \zeta_1 +a_2 \zeta_2
$$
where $\zeta_a$ are the degenerate Killing vectors of the NHG generated by $\varsigma_a\in \Kill_{S_0}$. Therefore, $\zeta '$ is itself a degenerate Killing vector field of the NHG. The only if part of theorem \ref{th:sexact} implies then that
$Z$ must satisfy
\begin{align*}
\frac{1}{Z}= \frac{Q'(f')}{\gamma^{\sharp}(df',df')},
\end{align*}
which was the only remaining condition for $\varsigma' \in \Kill_{S_0}$.
\end{proof}
\begin{corollary}\label{fsindependent}
For any pair of linearly independent $\varsigma_1, \varsigma_2 \in \Kill_{S_0}$, the corresponding functions $f_1$ and $f_2$ are functionally independent everywhere on $S_0$: $df_1\wedge df_2 \neq 0$.
\end{corollary}
\begin{proof}
From the proposition we know that $f'=a_1 f_1 + a_2 f_2$ gives rise to a degenerate Killing vector of type (\ref{degKill}) for the MKH  $\H_\eta$ of $(M_{\mbox{{\tiny NHG}}},g_{\mbox{{\tiny NHG}}})$. Theorem \ref{th:sexact} then tells us that $df'$ cannot vanish anywhere on $S_0$, and thus $a_1 df_1 + a_2 df_2 \neq 0$ everywhere on $S_0$. In particular, $df_1\wedge df_2 \neq 0$ on the whole $S_0$.
\end{proof}

\section{Structure of near-horizon geometries with MKHs of order $m\geq 3$}
\label{Sect:structure}
The following theorem relates the dimension of $\Kill_{S_0}$ to the
order of the MKH $\Htwo$.
\begin{theorem}
Let $(M_{\mbox{{\tiny NHG}}},g_{\mbox{{\tiny NHG}}})$
be a near horizon geometry with metric (\ref{NHG}). The (necessarily multiple)
Killing horizon $\Htwo$ has order $m \geq 2$ if and only if $\Kill_{S_0}$
has dimension $m-2$.
\end{theorem}

\begin{proof}
The proof is a simple combination of previous results. Let $\varsigma_i$, $i=1, \cdots ,p := \mbox{dim} \Kill_{S_0}$
be a basis of $\Kill_{S_0}$ and write $\bm\varsigma_i =
Z^{-1} df_i$. 
From corollary \ref{fsindependent} follows that
the $p+1$ functions $\{ 1, f_1, \cdots f_p\}$ are functionally independent.
For each such function we construct a vector field  in
$(M_{\mbox{{\tiny NHG}}},g_{\mbox{{\tiny NHG}}})$ according to (\ref{degKill}). By
proposition \ref{prop:KillingEqs} and theorem \ref{th:sexact} this yields a collection of
$p+1$ 
linearly independent degenerate Killing generators of $ \Htwo$. This implies that
$p \leq m-2$ where $m$ is the order of this multiple Killing horizon.
The reverse inequality is proved similarly starting with a set of $m-1$
linearly independent degenerate Killing generators of $\Htwo$.
\end{proof}

As the maximum order of a MKH is given by the spacetime dimension $n+1$, we have the following
\begin{corollary}
$\mathrm{dim}(\Kill_{S_0}) \leq n-1$.
\end{corollary}

In view of lemma \ref{ODESol}, expressions (\ref{sh}) for $h$ and $\bm{s}$  are
fully explicit. This allows us to find all possible NHGs with MKHs of order $m\geq 3$ in an explicit way 
locally. We establish the following fundamental result.
\begin{theorem}\label{th:warped}
Let $(M_{\mbox{{\tiny NHG}}},g_{\mbox{{\tiny NHG}}})$ be a near horizon geometry  with metric (\ref{NHG}), and assume that the  Killing horizon $\H_{\eta}$ is multiple with  order $m:=2+p\geq 3$. Let  $\varsigma_i\in \Kill_{S_0}$ ($i,j,k=1,\dots ,p$) be a set of linearly independent Killing vectors on any cut $S_0\subset \Htwo$ which give rise, together with $\eta=\partial_v$ and via Theorem \ref{th:sexact}, to $\Kill_{\Htwo}^{deg}$. For each $i$ let $Z$, $f_i$ and $Q_i(f_i)$ be functions on $S_0$ such that, with $N_i := \gamma^{\sharp}(df_i,df_i)$, $N_i = Q_i Z$ and $\bm\varsigma_i = Z^{-1} df_i$ hold. Then,
$\bm{s} = 1/(2Z) dZ$  and one of the two following mutually exclusive cases holds:
\begin{align*}
   (i) \quad  h =0 \quad \quad & \Longleftrightarrow \quad \forall i \quad
\begin{matrix}
\text{(a)} & Q_i(f_i ) = Q_{0i}
\\
& \mbox{\underline{or}}
\\
\text{(b)} &
Q_i(f_i)  = Q_{0i} ( f_i + c_i)
\end{matrix}
\quad \quad   Q_{0i}\neq 0, c_i \in \mathbb{R}  \\
 (ii) \quad  h = Z \quad \quad & \Longleftrightarrow \quad \forall i \quad
Q_i = Q_0 \left (b_i + (f_i + c_i)^2 \right ), \quad \quad \quad  \quad \quad
\quad   Q_0 \neq 0, c_i, b_i \in \mathbb{R}.
\end{align*}
In either case
$(S_0,\gamma)$ is (locally) a warped product $S_0=V\times \S$ with metric
\be
\gamma = \bar\gamma + \Omega g_{\varepsilon}, \hspace{1cm} \Omega: V \rightarrow \mathbb{R} \label{warped}
\ee
such that $(\S,g_{\varepsilon})$ is a $p$-dimensional maximally symmetric Riemannian manifold of constant curvature $\varepsilon$.
\end{theorem}
\begin{proof}
From lemma \ref{ODESol} we know that if $h=0$, then all  Killing vectors
$\varsigma_i$ must belong to either case (a) or case (b). This yields the expression in item (i)
of the theorem. Observe that in this item (i), if $\varsigma_k\in \Kill_{S_0}$ and $\varsigma_l\in \Kill_{S_0}$ belong to cases (a) and (b) respectively, then
$$
Z =\frac{N_k}{Q_{0k}} = \frac{N_l}{Q_{0l}(f_l+c_l)}
$$
must hold.

If $h$ is not identically zero, then all
$\varsigma_i$ must belong to case (c) in the same lemma. Thus $Q_i =Q_{0i} (b_i + (f_i + c_i)^2)$ for different in principle $Q_{0i}$.
Moreover $N_i = Z Q_i$ combined with
(\ref{expshc}) shows that $h=Z$ in this case.

To show that $\gamma$ is a warped product metric with
fibers of constant curvature. We start by showing that $\Kill_{S_0}$ is in involution,
i.e. that the commutators $[\varsigma_i, \varsigma_j]$ are linear combinations
 (with functions) of $\{ \varsigma_i \}$. Note first that from
(\ref{sDf})
\begin{align}
D^A f_i D_A Z = 2 Z (D^A f_i) s_A = Z P_i N_i, \label{intermediate1}
\end{align}
and from the master equation (\ref{master})
\begin{align}
D^A f_i D_A D_B f_j = \frac{1}{2Z}  D^A f_i \left (
D_A Z D_B f_j  + D_B Z D_A f_j \right )
=
\frac{1}{2} P_i N_i D_B f_j
+ \frac{D_B Z}{2Z} D^A f_i D_A f_j
\label{intermediate2}
\end{align}
Given that  $\bm\varsigma_i = Z^{-1} d f_i$
we compute
\bea
[ \varsigma_i, \varsigma_j]^B & = &
\frac{1}{Z}D^Af_i D_A\left(\frac{1}{Z} D^Bf_j \right) - (i \longleftrightarrow j)\nonumber \\
& = & \frac{1}{Z^2} D^A f_i \left ( - \frac{1}{Z} D_A Z D^B f_j +
D_A D^B f_j \right )- (i \longleftrightarrow j)\nonumber \\
 &= &\frac{1}{Z^2} \left ( - \frac{1}{2} N_i P_i D^B f_j +
\frac{1}{2} N_j P_j D^B f_i \right )  \nonumber\\
& = & \frac{1}{2} \left ( Q_j P_j \varsigma_i^B -
Q_i P_i \varsigma_j^B\right )\label{comm}
\eea
where in the third equality we inserted (\ref{intermediate1}) and
(\ref{intermediate2}).

We next  recall that, according to corollary \ref{fsindependent},
$\{ \varsigma_i \}$ are not only linearly independent as vector fields, but even more, linearly independent at every point  $q \in S_0$, so that the vector space
$T_q := \mbox{span} \{ \varsigma_i |_q \}$  has dimension $p$.
Thus, the collection $\{T_q\}$
defines a distribution of dimension $p$
which is in involution. By the Fr\"obenius theorem, $S_0$ can be foliated by
injectively immersed integrable manifolds of dimension $p$, whose
tangent space is $T_q$. Since we work locally on $S_0$, we may assume that
these integral manifolds are embedded. We want to show that the induced metric
is of constant curvature. Equivalently, we must show that the integral manifolds are maximally symmetric.  To that aim define
\begin{align*}
\varsigma_{ij} := f_j \varsigma_i - f_i \varsigma_j,
\end{align*}
which are clearly tangent to the integral manifolds. It turns out that
$\varsigma_{ij}$ are Killing vectors of
$(S_0,\gamma)$, as follows from
\bean
D_A (\varsigma_{ij})_B = D_A f_j (\varsigma_i)_B + f_j D_A (\varsigma_i)_B- D_A f_i (\varsigma_j)_B - f_i D_A (\varsigma_j)_B\\
=\frac{1}{Z}\left(D_A f_j D_Bf_i - D_A f_i D_Bf_j \right)+ f_j D_A (\varsigma_i)_B-f_i D_A (\varsigma_j)_B=D_{[A} (\varsigma_{ij})_{B]}.
\eean
Moreover the Killing vectors $\{ \varsigma_i, \varsigma_{ij} (i< j) \}$ are linearly
independent:  take constants $\{ \alpha^i, \beta^{ij} =
\beta^{[ij]}\}$ satisfying
\begin{align*}
0 = \alpha^i \bm\varsigma_i + \beta^{ij} \bm\varsigma_{ij} =
\frac{1}{Z} \left ( \alpha^i + 2 \beta^{ij} f_j \right )  df_i.
\end{align*}
Since $df_i$ are linearly independent at every point it must be
$\alpha^i + 2 \beta^{ij} f_j =0$. But remember that the set $\{ 1, f_i \}$
is functionally independent from where we conclude that $\alpha^i = \beta^{ij} =0$ proving that $\{ \varsigma_i, \varsigma_{ij} (i< j) \}$ is a set of linearly independent Killing vector fields.
Hence, we have
$$p+\frac{p(p-1)}{2} = \frac{p(p+1)}{2}
$$
linearly independent Killing vectors
$\{\varsigma_i,\varsigma_{ij}\}$. Since they are all tangent everywhere to
$p$-dimensional manifolds, these spaces are maximally symmetric
(and therefore, of constant curvature).
Moreover, it must be the case that
$\{\varsigma_i,\varsigma_{ij}\}$ generate a Lie subalgebra of the Killing Lie algebra of $(S_0,\gamma)$. By a theorem due to Schmidt \cite{Sch}, or by noticing that $\{\varsigma_i\}_{i=1,\dots,p}$ generate orthogonal hypersurfaces, the orbits admit orthogonal submanifolds, and thus the metric decomposes as the warped product (\ref{warped}).

It remains to show that, in item (ii), all constants $Q_{0i}$ are equal to each other. This follows from the commutator (\ref{comm}) because in case (c) of lemma \ref{ODESol} one has
$$
[ \varsigma_i, \varsigma_j]=Q_{0j} (f_j+c_j)\varsigma_i - Q_{0i} (f_i+c_i)\varsigma_j
$$
which can be rewritten as
$$
[ \varsigma_i, \varsigma_j]+Q_{0i} c_i\varsigma_j -Q_{0j} c_j\varsigma_i -Q_{0j}\varsigma_{ij} =(Q_{0j}-Q_{0i})f_i \varsigma_j \, .
$$
The lefthand side here is a Killing vector field, and thus the righthand side must be too. But given that $\varsigma_j$ is a Killing vector itself and that $df_i\neq 0$, this can only happen if $Q_{0i}=Q_{0j}:=Q_0$, finishing the proof.
\end{proof}

The constant curvature of $(\Sigma,g_{\varepsilon})$ can be explicitly found.
\begin{theorem}\label{th:signofK}
Under the same hypotheses of theorem \ref{th:warped}, if $p> 1$ the constant sectional curvature $\varepsilon$ of the fibers $(\Sigma,g_{\varepsilon})$ is given by
\begin{enumerate}
\item $\varepsilon = - \G^{kl}q_k q_l $
in item (i) of theorem \ref{th:warped} where
\be
q_i := \left\{
\begin{array}{cc}
0 & \mbox{if $ \varsigma_i$ belongs to case (a)}\\
\frac{1}{2}Q_{0i} & \mbox{if $\varsigma_i$ belongs to case (b)}
\end{array}
\right.
\label{qi}
\ee
\item $\varepsilon =Q_0\left\{Z-Q_0\G^{kl} (f_k+c_k)(f_l+c_l)\right\}$ in item (ii) of theorem \ref{th:warped}
\end{enumerate}
where ${\cal G}^{kl}$ are the components of $g_\varepsilon^\sharp$ in the basis $\{\varsigma_i\}$.
\end{theorem}
\begin{proof}
The curvature tensor of the orbits of a group of motions can be computed in terms of its Lie algebra structure constants and the metric, as proven in \cite{Sch1}, see also \cite{Exact} section 8.6. We apply a slightly simpler modification of that calculation to our situation. First of all, note that the set of hypersurface-orthogonal Killing vectors $\{\varsigma_i\}$ defines a frame, that is, a basis of $\mathfrak{X}(S_0)$. Then, we wish to compute the curvature in this basis. The components of the metric in this basis will be denoted by
$$
\G_{ij}:= g_\varepsilon (\varsigma_i,\varsigma_j), \hspace{1cm} \G^{ki}\G_{ij}=\delta^k_j .
$$
We start by computing the Levi-Civita connection $\n$ of $(\Sigma,g_\varepsilon)$ in this basis, that is, $\n_{\varsigma_i} \varsigma_j$.
 Its anti-symmetric part in $ij$ is given by the commutator (\ref{comm}). For the symmetric part, we use a formula derived in \cite{Sch1,Exact} for arbitrary Killing vectors
$$
g_\varepsilon \left(\varsigma_l, \n_{\varsigma_i} \varsigma_j+\n_{\varsigma_j} \varsigma_i\right)=
g_\varepsilon \left(\left[ \varsigma_i,\varsigma_l\right],\varsigma_j\right)+g_\varepsilon \left(\left[ \varsigma_j,\varsigma_l\right],\varsigma_i\right)=Q_l P_l \G_{ij} -\frac{1}{2} Q_i P_i \G_{lj} -\frac{1}{2} Q_j P_j \G_{li}
$$
where in the last equality we have used (\ref{comm}). Adding the symmetric and antisymmetric parts we deduce (summation on $l$ and $k$ is understood)
\be
2 \n_{\varsigma_i} \varsigma_j = \G_{ij} Q_l P_l \G^{lk}\varsigma_k - Q_i P_i \varsigma_j . \label{Dss}
\ee
To compute the second derivative $\n_{\varsigma_k}\n_{\varsigma_i} \varsigma_j$ using this formula one needs to know the derivatives of $\G_{ij}$ and $\G^{ij}$ along $\varsigma_k$. But these are easily found on using again (\ref{comm}) as
\bean
\varsigma_i (\G_{jk}) &=& \varsigma_i (g_\varepsilon (\varsigma_j,\varsigma_k))= g_\varepsilon (\pounds_{\varsigma_i }\varsigma_j, \varsigma_k) +g_\varepsilon (\varsigma_j ,\pounds_{\varsigma_i }\varsigma_k)= \frac{1}{2} Q_j P_j \G_{ik} +\frac{1}{2} Q_k P_k \G_{ij}-  Q_i P_i\G_{jk}, \\
\varsigma_i (\G^{lm}) &=& - \G^{lk} \G^{mj} \varsigma_i (\G_{jk}).
\eean
The only remaining derivatives to be known are those of $Q_j P_j$, but for these it is convenient to separate the cases and consider their explicit form.
\begin{itemize}
\item For item (i) in the theorem, cases (a) and (b), we have $Q_i P_i =2 q_i $ are constants according to definition (\ref{qi}), and thus $\varsigma_j (Q_i P_i ) = 2 \varsigma_j (q_i) =0$. Putting everything together and after a little calculation using (\ref{Dss}) one can then easily obtain
$$
\n_{\varsigma_i} \n_{\varsigma_j} \varsigma_k - \n_{\varsigma_j} \n_{\varsigma_i} \varsigma_k-
\n_{[\varsigma_i,\varsigma_j]} \varsigma_k = -(\G^{mn}q_m q_n ) \left(\G_{jk} \varsigma_i -\G_{ik} \varsigma_j\right)
$$
which proves 1. (It is easily seen with the used formulas that $\varsigma_i(\G^{mn}q_m q_n)=0$ so that $\G^{mn}q_m q_n$ is actually constant, as it must).
\item For item (ii) in the theorem, that is for case (c), we have $Q_j P_j = 2Q_0(f_j+c_j)$ so that we need $\varsigma_i(f_j)$. This can be derived from the following calculation
\bean
&&d(\varsigma_i (f_j)) =\pounds_{\varsigma_i} (df_j) = \pounds_{\varsigma_i} (Z\bm\varsigma_j)=\varsigma_i (Z) \bm\varsigma_j +Z\pounds_{\varsigma_i}\bm\varsigma_j=\\
&&2Q_0(f_i+c_i)Z\bm\varsigma_j +Z Q_0\left\{(f_j+c_j) \bm\varsigma_i -(f_i+c_i)\bm\varsigma_j\right\}=
Q_0  \left\{(f_i+c_i)df_j + (f_j+c_j)df_i \right\}
\eean
where (\ref{comm}) and
\be
\varsigma_i(Z)=2Q_0 Z (f_i+c_i) \label{DZ}
\ee
which comes from (\ref{intermediate1}) have been used. This gives immediately
\be
\varsigma_i(f_j) =Q_0 \left\{(f_i+c_i)(f_j+c_j)+c_{ij} \right\}, \hspace{1cm} c_{ij}\in \mathbb{R} \label{sfj}
\ee
where the constants $c_{ij}$ are such that $c_{ii} =b_i$. Observe that
$$
\G_{ij} =g_\varepsilon (\varsigma_i, \varsigma_j) =Z^{-1} \varsigma_i(f_j) =Z^{-1}Q_0 \left\{(f_i+c_i)(f_j+c_j)+c_{ij} \right\},
$$
which is particular implies  the symmetry $c_{ij} = c_{ji}$.
Using this and putting everything together, a little longer calculation leads easily to
$$
\n_{\varsigma_i} \n_{\varsigma_j} \varsigma_k - \n_{\varsigma_j} \n_{\varsigma_i} \varsigma_k-
\n_{[\varsigma_i,\varsigma_j]} \varsigma_k =
Q_0\left\{Z-Q_0\G^{kl} (f_k+c_k)(f_l+c_l)\right\}\left(\G_{jk} \varsigma_i -\G_{ik} \varsigma_j\right).
$$
It is a matter of checking that the quantity $Q_0\left\{Z-Q_0\G^{kl} (f_k+c_k)(f_l+c_l)\right\}$ has zero derivative along all $\varsigma_i$ on using (\ref{DZ}) and (\ref{sfj}). This finishes the proof.
\end{itemize}
\end{proof}

For completeness, we now check that all Killing vectors of the form (\ref{degKill}) related to elements $\varsigma \in \Kill_{S_0}$ according to theorem \ref{th:sexact}
necessarily commute with each other and with $\partial_v$,
in agreement with theorem \ref{th:algebra}. To that end, take any two such Killing vectors
$$
\zeta_a = f_a\partial_v +\frac{u^2}{2} \Delta f_a \partial_u - u\,  {\rm grad} f_a
$$
with $a\in\{1,2\}$. Using that $Z \varsigma_a ={\rm grad} f_a$ for $Z$ as defined above, so that $\varsigma_1(f_2)=\varsigma_2(f_1)$ and $\Delta f_a =\varsigma_a (Z)$, a direct computation provides for their commutator
$$
\left[\zeta_1,\zeta_2 \right]= \frac{u^3}{2} Z \left\{\varsigma_2 (\Delta f_1) - \varsigma_1(\Delta f_2)\right\}\partial_u+u^2 Z \left\{Z\left[\varsigma_1,\varsigma_2 \right] +\frac{1}{2} \varsigma_1(Z)\varsigma_2 -\frac{1}{2}\varsigma_2(Z)\varsigma_1\right\}.
$$
The term proportional to $u^2$ vanishes due to (\ref{intermediate1}) and (\ref{comm}). Concerning the first term, proportional to $u^3$, we note that $\varsigma_a$ are Killing vectors of $(S_0,\gamma)$, and thus they commute with the Laplacian $\Delta$, which easily leads to $\left[\zeta_1,\zeta_2 \right]=0$, as expected.

\section{Uniqueness of NHG for MKHs of order $m\geq 3$}
\label{Uniqueness}
Once we know the precise explicit form of the possible NHGs with MKHs of order $m\geq3$ we are ready to address the problem whether or not there can be several distinct NHGs arising from a MKH $\H_{\mathrm{MKH}}$  which admits at least two independent degenerate Killing vectors.
More precisely, let $\eta,\zeta\in \Kill_{\H_{\mathrm{MKH}}}^{deg}$ denote  degenerate Killing vectors. We want to analyze whether there is a (local) isometry between the
NHGs of $\eta$ and $\zeta$ associated to (each connected component of)
the horizon $\H_{\eta,\zeta}:=\H_{\eta}\cap\H_{\zeta}\subset \H_{\mathrm{MKH}}$.

Let $(M_{\mbox{{\tiny NHG}}},g_{\mbox{{\tiny NHG}}})$ be a near horizon geometry with metric (\ref{NHG}), and assume that the Killing horizon $\Htwo =\{u=0\}$ is multiple with order $m\geq 3$. Then its Lie algebra $\Kill_{\Htwo}$ includes at least the following Killing vector fields
$$
\xi =v\partial_v -u\partial_u, \quad \eta =\partial_v, \quad \zeta = f\partial_v +\frac{u^2}{2} \Delta f \partial_u -u \, \mbox{grad} f
$$
where the last one satisfies the relations proven in theorem \ref{th:sexact}, so that in particular $f$ is a solution of (\ref{master}). The NHG of this $(M_{\mbox{{\tiny NHG}}},g_{\mbox{{\tiny NHG}}})$ with respect to $\eta$ is obviously itself, as follows from the intrinsic characterizations of $h$, $\bm{s}$ and $\gamma$ given in Remark \ref{intrinsic}. However, one can also construct another NHG for $(M_{\mbox{{\tiny NHG}}},g_{\mbox{{\tiny NHG}}})$ by using as degenerate Killing $\zeta$
and by restricting the horizon to $\H_{\eta,\zeta}$. It is a matter of checking that, for this Killing $\zeta$, the corresponding intrinsic elements of Remark \ref{intrinsic} characterizing its NHG (call them $\bm{S}, H$ and $\tilde \gamma$) are such that $\tilde\gamma$ is isometric to $\gamma$ (and thus we will remove the tilde when not using a coordinate system) and the others are given by
\be
\left . \bm{S} \right |_{S_0} =
\left . \bm{s} - \frac{df}{f}  \right |_{S_0}, \quad  \quad
\left .  H \right  |_{S_0}
= \left . h- \frac{\Delta f}{f} +\frac{N}{f^2} \right |_{S_0} .
\label{SH}
\ee
Note that $f$ is non-zero on $\H_{\eta,\zeta}$.

By the standard near-horizon construction, there exist
coordinates $\{U,V,y^A\}$ such that
\be
g_{\mbox{{\tiny NHG}}}(\zeta)
= 2dV(dU +2U S_A dy^A +\frac{1}{2} U^2 H dV) +\tilde\gamma_{AB} dy^A dy^B. \label{NHGy}
\ee
and $H$,  $\bm{S}$ and $\tilde{\gamma}$ have been extended off $S_0$ as functions independent of $U$ and $V$.

This information is enough to prove the preliminary result that
$g_{\mbox{{\tiny NHG}}}(\zeta)$ admits a degenerate Killing generator for which
the corresponding near horizon geometry brings us back to the original metric.
\begin{proposition}\label{prop:noiter}
The MKH $\H := \{U=0\}$ of the metric $g_{\mbox{{\tiny NHG}}}(\zeta)$ has at least three  Killing vectors given by
$$
\tilde\xi =V\partial_V -U\partial_U, \quad  \tilde\eta =\partial_V,
\quad \tilde\zeta = \frac{1}{\f}\partial_V +\frac{U^2}{2} \Delta
\left (\frac{1}{\f} \right )\partial_U -U \mbox{{\rm grad}} \frac{1}{\f},
$$
where  $\tilde\eta$ and $\tilde\zeta$ are degenerate
and the function $F$ is independent of $U$ and $V$ and satisfies $\f |_{S_0} = f$.
Moreover, the NHG of
$\H_{\tilde \eta,\tilde\zeta}\subset \H$ with respect to $\tilde\zeta$ is the original metric $g_{\mbox{{\tiny NHG}}}$ in (\ref{NHG}).
\end{proposition}
\begin{proof}
  From theorem \ref{th:sexact} we know that the Killing vectors of $g_{\mbox{{\tiny NHG}}}(\zeta)$ which are degenerate at $\H$ other than $\partial_V$ must solve equations (\ref{master}-\ref{eq2}) where now $h$ and $\bm{s}$ must be substituted by their corresponding capital-letter versions. But given that $f$ solves these equations (\ref{master}-\ref{eq2}) in their original form, it is easy to see that $1/f$ solves the new equations. It suffices to work on $S_0$, where
$F= f$
\bean
&&D_A D_B (1/f) -S_A D_B(1/f)  -S_B D_A (1/f) \\
&&= -\frac{1}{f^2} D_A D_B f +\frac{2}{f^3} D_A f D_B f  -\left(s_A -\frac{1}{f} D_Af\right) \left(-\frac{D_Bf}{f^2}\right)
-\left(s_B -\frac{1}{f} D_Bf\right) \left(-\frac{D_Af}{f^2}\right) \\
&& = -\frac{1}{f^2}\left(D_A D_B f -s_A D_B f-s_B D_A f \right) =0.
\eean
Similar, but a little longer calculations, using (\ref{master}-\ref{eq2}) prove that
\bean
&& D^A (1/f) \left( D_A H - 2 S_A H\right)=0, \\
&& H D_A (1/f)  = 2D^B (1/f) \left(D_B S_A -D_A S_B \right)+D_A (S^B D_B (1/f)) - 2 S_A S^B D_B (1/f).
\eean
In simpler words, $1/f$ is a solution of the corresponding equations, leading to the last Killing in the list given in the Proposition. But this easily implies that repeating the NHG process just leads to the original metric given in (\ref{NHG}).
\end{proof}

We want to analyze whether
(\ref{NHG})  and (\ref{NHGy}) are isometric to each other or not.  Our strategy
will be to assume that they are (locally) isometric and find a set of necessary conditions that need to be satisfied, in particular regarding the explicit form
of the coordinates
$\{U,V,y^A\}$ in terms of $\{u,v,x^A\}$. We will then confirm
that this coordinate change indeed transforms one metric into the other.

Assume thus that the two spacetimes are locally isometric and that the isometry takes the cross section $\{u=v=0\}$ into the cross section
$\{ U =V =0\}$. The only Killing vector of
$g_{\mbox{{\tiny NHG}}}$ which is (i) a Killing generator
of $\Htwo$, (ii) vanishes on $S_0$  and (iii) has surface gravity $\kappa =1$
is $\xi$. Similarly, the only Killing vector of
$g_{\mbox{{\tiny NHG}}}(\zeta)$  satisfying the  corresponding properties is $\tilde{\xi}$. Thus, we are forced to identify $\xi$ with $\tilde{\xi}$.

The vector field
$\partial_V$ in  $g_{\mbox{{\tiny NHG}}}(\zeta)$ is by construction the Killing vector with respect to which we have performed the near horizon limit. Since in the
original coordinates this vector is $\zeta$, we are led to identify $\zeta$ with
$\tilde{\eta}$. Concerning $\eta$, this vector is a degenerate Killing generator
of $\Htwo$. By the assumed isometry, it must also be a degenerate Killing
generator of
$\H_{\tilde\eta,\tilde\zeta}\subset\H = \{ U=0\}$. Let $m$ be the order of this multiple Killing horizon (we know that $m \geq 3$ by Proposition \ref{prop:noiter})
 and $\{ \tilde{\eta}
, \tilde{\zeta}, \tilde{\zeta}_a \}$ ($a=2, \cdots m-1$) be a basis of degenerate
Killing  generators of $\H_{\tilde\eta,\tilde\zeta}$.
The assumed isometry forces the existence of constants $\alpha, \beta, \beta^a$ such that
\begin{align}
\eta =
\alpha \tilde{\eta}
+ \beta \tilde{\zeta} + \beta^a \tilde{\zeta_a}.
\label{decom}
\end{align}
We now use the proportionality  $\zeta |_{S_0} = f \eta |_{S_0}$ and let
$\tilde{f}_a$ be the functions on $S_0$ defined by $\tilde{\zeta}_a |_{S_0}
= \tilde{f}_a \partial_V$.  Thus,
\begin{align*}
\tilde\eta |_{S_0} =\partial_V |_{S_0} = \zeta |_{S_0} = f \eta |_{S_0}
= f \left ( \alpha \tilde{\eta}
+ \beta \tilde{\zeta} + \beta^a \tilde{\zeta_a} \right ) |_{S_0} =
f \left ( \alpha + \frac{\beta}{f} + \beta^a \tilde{f}_a \right )
\partial_V  |_{S_0} \, .
\end{align*}
Since $\{1, 1/f, \tilde{f}_a \}$ are functionally independent solutions of
the master equation (\ref{master}) in the metric
$g_{\mbox{{\tiny NHG}}}(\zeta) $, we conclude that the only possibility is
$\alpha = 0, \beta =1$ and $\beta^a=0$ and we are forced
to identify $\eta$ and $\tilde{\zeta}$,
  and thus  $\H_{\eta,\zeta}$ and $\H_{\tilde \eta,\tilde\zeta}$.

We will have to deal with functions that agree on $S_0$ but are
extended in two different ways off $S_0$, namely as functions
independent of $\{U,V\}$ and as functions independent of $\{u,v\}$.
A more geometric way to state this is that the set of points at equal value of $x^A$
are different from the set of points  at equal value of $y^A$, even if these
coordinates agree on $S_0$. The corresponding functions obtained by the two
extensions are different and hence
must be given different names.
An example is the pair of functions $f$ and $F$, which agree on $S_0$ but have been
extended so that $f$ is independent of the coordinates $\{u,v\}$ while
$F$ is independent of the coordinates $\{U,V\}$. Note that this has nothing to
do with the fact that either $F$ or $f$ can still be expressed in any
coordinate system one wishes. The function $h$ is independent of $\{u,v\}$.
The function that agrees with $h$ on $S_0$
  but is extended as a function independent of $\{U,V\}$ will be denoted by
  $\hat{h}$. Similarly, we define $\hat{H}$ as the function
  that agrees with $H$ on $S_0$ and is extended as independent of $\{u,v\}$.
  From the second in (\ref{SH}) one has
    \begin{align*}
      \hat{H} = h - \frac{\Delta f}{f} + \frac{N}{f^2}.
    \end{align*}
Our first step in the process of determining the coordinate change is to
 impose that the scalar products of various Killing fields must agree when computed with respect to each metric. Specifically it must be that
    $\gNH (\zeta) (\tilde{\eta}, \tilde{\eta}) =
\gNH (\zeta, \zeta)$, which after a simple calculation that uses
(\ref{Deltaf}) provides
\begin{align}
  U^2 H =  u^2 f^2 \left ( h - \frac{\Delta f}{f} + \frac{N}{f^2} \right )
  =  u^2 f^2 \hat{H} \label{nose}.
\end{align}
Similarly, the equality
$\gNH (\zeta) (\tilde{\zeta}, \tilde{\zeta}) =
\gNH (\eta, \eta)$ yields
\begin{align}
  u^2 h = \frac{U^2 \hat{h}}{F^2}.
  \label{hh}
\end{align}
The equalities
$\gNH(\xi,\xi) = \gNH(\zeta) (\tilde{\xi}, \tilde{\xi})$ and
$\gNH(\xi,\eta) = \gNH(\zeta) (\tilde{\xi}, \tilde{\zeta})$ yield, after
another simple computation
\bea
uv(uv \,h -2) &=& UV(H UV-2) , \label{eqa}\\
u(uv \, h -1) &=& U^2V \left(\frac{H}{F} -\frac{1}{2} \Delta
\left (\frac{1}{F}\right ) \right ) -\frac{U}{F}. \label{eqb}
\eea
Let us next find equations that must be satisfied by the change of coordinates
between $\{u,v,x^A\}$ and $\{U,V,y^A\}$.
The identification of $\tilde{\eta}$ and
$\zeta$ leads to
\begin{align}
  \tilde{\eta} (v) &= \frac{\partial v}{\partial V} = \zeta(v) = f,
  \label{der1} \\
  \tilde{\eta} (u) &= \frac{\partial u}{\partial V} = \zeta(u) = \frac{1}{2}
  u^2 \Delta f,  \label{der2} \\
  \tilde{\eta} (x^A) &= \frac{\partial x^A}{\partial V} = \zeta(x^A) =
  -u \, \mbox{grad} f. \label{dera}
\end{align}
Similarly, the identification of $\xi$ and $\tilde{\xi}$ implies
\begin{align}
  v &= \xi(v) = \tilde{\xi} (v) = V \frac{\partial v}{\partial V}
  - U \frac{\partial v}{\partial U}, \label{der3} \\
  u &= - \xi(u) = \tilde{\xi} (v) = - V \frac{\partial u}{\partial V}
  + U \frac{\partial u}{\partial U}, \label{der4} \\
0 &= \xi(x^A) = V \frac{\partial x^A}{\partial V}
- U \frac{\partial x^A}{\partial U}. \nonumber
\end{align}
The last equation gives
\begin{align}
  x^A = X^A (UV, y).
  \label{changexA}
\end{align}
Let us integrate the pair (\ref{der1})-(\ref{der3}). Inserting
the second into the first yields the equivalent system
\begin{align*}
  U \frac{\partial v}{\partial U} + v = V f, \quad \quad
  \frac{\partial v}{\partial V} = f.
\end{align*}
Defining $C(U,V,y)$ by $\displaystyle{v = \frac{C}{U}}$ this system becomes
\begin{align*}
  \frac{\partial C}{\partial U} = V f(X(UV,y)),
  \quad \quad \quad
\frac{\partial C}{\partial V} = U f(X(UV,y)),
\end{align*}
As a consequence $U \partial_U C - V \partial_V C=0$ and  hence $C(UV,y)$.
The general solution is therefore
\begin{align}
  v = \frac{1}{U} C(UV,y), \quad \quad \frac{\partial C}{\partial (UV)}
  = f(X(UV,y)). \label{C}
\end{align}
We next solve the pair (\ref{der2})-(\ref{der4}).  In terms of the function
$G(U,V,y)$ defined by $u = U/G$, the system becomes
\begin{align*}
  \frac{\partial G}{\partial V} = - \frac{1}{2} U \Delta f,
  \quad \quad
  0 = V \frac{\partial G}{\partial V} - U \frac{\partial G}{\partial U},
\end{align*}
so that its general solution is
\begin{align}
u= \frac{U}{G(UV,y)}, \quad \quad \frac{\partial G}{\partial (UV)} =-\frac{1}{2} \Delta f (X(UV,y)).\label{G}
\end{align}
Observe that
$$
uv =\frac{C}{G} (UV,y).
$$
Eq.(\ref{eqa}) together with this provides an expression for $h$
\be
h  = \frac{G}{C}\left(\frac{G}{C} UV(UV \, H-2)+2\right)\label{hx}
\ee
and the combination of (\ref{hh}) with (\ref{G}) another one
\be
h = G^2 \frac{\hat{h}}{F^2} \label{hx1}.
\ee
The combination of (\ref{hx}) with (\ref{hx1}) provides a relation involving only coordinates on one side
\be
GC\frac{\hat{h}}{F^2} =\frac{G}{C} UV ( UV \, H-2)+2 \label{GC}.
\ee
Combining (\ref{eqb}) with (\ref{hx}) gives another such relation
\be
\frac{1}{G} = C \frac{\hat{h}}{F^2} -UV \left( \frac{H}{F}
-\frac{1}{2} \Delta \left (\frac{1}{F} \right )\right)+\frac{1}{F}
\label{G1}.
\ee
From (\ref{GC}) and (\ref{G1}) we can thus get the two functions $C$ and $G$, given by
\bea
G&=&\frac{F}{\Xi},\label{GXi}\\
C&=&\frac{F}{\hat{h}}\left[UV\left( \hat{h}-\frac{1}{2}\frac{\Delta F}{F}\right)-1+\Xi  \right] \label{CXi}
\eea
where we have defined the abbreviation
$$
\Xi:= \sqrt{\left(1+UV\frac{\Delta F}{2F} \right)^2 -U^2V^2 \frac{\hat{h}}{F^2}
      N_F}.
$$
where $N_F$ is the squared norm of $dF$ in the metric $\tilde{\gamma}$
of (\ref{NHGy}).
    \begin{remark}
  Expression (\ref{CXi}) seems to have a problem if $\hat{h}=0$, i.e.
  when the starting metric $\gNH$ has $h=0$.
  However, it is easy to check that (\ref{CXi}) has a well defined limit when
  $\hat{h}\rightarrow 0$, given by
\be
\hat{h}=0 \hspace{3mm} \Longrightarrow \hspace{3mm} C= \frac{UV F}{1+UV\frac{\Delta F}{2F}}\left[1+UV\left(\frac{\Delta F}{2F}  -\frac{N_F}{2 F^2} \right) \right].\label{Ch=0}
\ee
\end{remark}

    All the previous formulas are given for general functions $h,H,f$ and for general metric $\gamma$ and one-forms $\bm{s}$ and $\bm{S}$. However, we know from theorem \ref{th:warped} that if the MKH in the NHG has order $m\geq 3$, then these objects take a very particular, explicitly known, form. Then, we have to take this into account and incorporate these explicit forms into the previous relations in order to find the sought isometry (coordinate change) between (\ref{NHG}) and (\ref{NHGy}).

    First of all we write down in a more explicit form the two spacetime metrics.
    The function $f$ has nowhere gradient on $S_0$ and, according to
    theorem \ref{th:warped} the metric
    $\gamma$ on $S_0$ takes the form
    \begin{align*}
      \gamma = \bar \gamma + \frac{1}{N} df\otimes df
      = \bar{\gamma} + \frac{1}{Q^2(f) M} df \otimes df
          \end{align*}
    where in the second equality we used the function $M$ defined
    by $N = Q^2 M$ which was introduced in  Remark \ref{remarkb}. Let $\{ x^{A'} \}$ ($A',B' \in \{ 2, \cdots n-1\}$
      be a local coordinate in the base manifold
      $V$ of the warped product $S_0 = V \times \Sigma$. Then $\{ x^A \}
      := \{ x^{A'}, f \}$ is a coordinate system of $S_0$. From Remark
      \ref{remarkb}, $M$ depends only on $x^{A'}$  and $M^{-1}$ is thus the warping function. Using the expressions for
      $\bm{s}$ and $h$ given in theorem \ref{th:sexact}, the metric $\gNH$
      takes the form ($\otimes_s$ denotes symmetrized tensor
    product)
      \begin{align}
        \gNH = 2 du dv + 2 u  d \ln (|Q M |) \otimes_s dv
         + \frac{1}{2} u^2 M Q^2 \left ( \frac{dP}{df} + P^2
         \right ) dv^2 + \bar\gamma + \frac{1}{Q^2 M} df^2 .
\label{firstmetric}
              \end{align}
      Concerning the metric
      $\gNH(\zeta)$, we first note that the expression of $\bm{S}$ and
      $H$ on $S_0$ are, according to (\ref{SH}) and recalling that $\Delta f
      = P(f) N = P Q^2  M$,
      \begin{align}
        \bm{S} |_{S_0} & = \frac{1}{2} \left ( Pdf + \frac{dM}{M} -
        2 \frac{df}{f} \right ), \nonumber \\
              H |_{S_0} & = Q^2 M \left ( \frac{1}{2} \left ( \frac{dP}{df} + P^2
        \right )
        - \frac{P}{f} + \frac{1}{f^2} \right )        := M K(f), \label{defK}
      \end{align}
      where the last expression defines the function of one variable $K(f)$.
      The gradient of $f$ in $S_0$ is clearly
      $\mbox{grad} f = Q^2 M \partial_f$
      and equation (\ref{dera}) for $A=A'$ together with
      (\ref{changexA})
      readily implies that $x^{A'}(y)$ (i.e. independent of $U,V$).
      Without loss of generality we can
      choose the coordinate system $\{y\}$ on $S_0$ to be the same
      as $\{x\}$. Hence $x^{A'} = y^{A'}$ everywhere. This means in particular
that the function $M$  does not change by the coordinate transformation and
      the following expression holds everywhere (because they hold on $S_0$ and both sides are functions independent of $\{U,V\}$)
               \begin{align*}
          N_F = M Q^2(F), \quad \quad
        \bm{S} = \frac{1}{2} \left ( P(F) dF
        + \frac{dM}{M} - 2 \frac{dF}{F} \right ),  \quad
        \quad H = M K(F)
      \end{align*}
        and the metric $\gNH(\zeta)$ is
      \begin{align}
        \gNH(\zeta) =
        2 dU dV
        + 2 U d \ln  \left ( \frac{| M Q(F)|}{F^2} \right ) \otimes_s dV
        + U^2 M K(F) dV^2 + \bar \gamma + \frac{dF^2}{M Q^2(F)}.
        \label{secondmetric}
      \end{align}
From this point on we need to distinguish between the
three possible cases according to  Lemma \ref{ODESol}, that is
cases (a) or (b) which satisfy  $h=0$, or case (c), for which $h \neq 0$.

\subsection{The case  (a)}
This case satisfies $P(f)=0$ and $Q(f) = Q_0$. By a trivial
rescaling of $M$ we may set $Q_0=1$ without loss of generality. The function
$K(f)$ in (\ref{defK}) is $K(f)= 1/f^2$ and the metrics
(\ref{firstmetric}) and (\ref{secondmetric}) to be compared become
\begin{align}
  \gNH & = 2 du dv + 2 u \, d \ln M \otimes_s dv
      + \bar\gamma_{A'B'} dx^{A'} dx^{B'}  + \frac{1}{M} df^2, \label{Fmeta}\\
    \gNH(\zeta) & =
        2 dU dV
        + 2 U d \ln  \left ( \frac{M}{F^2} \right ) \otimes_s dV
        + \frac{U^2 M}{F^2}  dV^2 + \bar \gamma_{A'B'} dy^{A'} dy^{B'}
        + \frac{dF^2}{M}. \label{Smeta}
\end{align}
Since $\Delta F=0$, $h=0$ (and hence $\hat{h}=0$), the function
$\Xi$ simplifies to $\Xi = 1$ so that (\ref{GXi}) and (\ref{Ch=0}) yield
\begin{align*}
  G= F,  \quad \quad  C = U V F \left ( 1 -  \frac{U V M}{2F^2} \right ).
\end{align*}
From (\ref{C})  and (\ref{G}) the coordinate change is
\begin{align}
  u = \frac{U}{F}, \quad \quad v = V F \left ( 1 - \frac{UV M}{2F^2} \right ),
  \quad \quad f = \frac{\partial C}{\partial (UV)} = F - \frac{U V M}{F},
  \quad \quad x^{A'} = y^{A'}. \label{isometry1}
\end{align}
A straightforward computation shows that this coordinate change
indeed transforms (\ref{Fmeta}) into (\ref{Smeta}).

\subsection{The case (b)}

In this case we have
\begin{align*}
  Q (f) = Q_0 (f+c), \quad \quad P(f) = \frac{1}{f+c}.
\end{align*}
Again a trivial rescaling of $M$ allows one to set $Q_0=1$. The function $K(f)$
is,  from (\ref{defK}),
\begin{align*}
  K(f) = \frac{c (f+c)}{f^2}
\end{align*}
so the two metrics to be compared are
\begin{align}
  \gNH  = &  2 du dv + 2 u \, d \ln ( M |f+c| )  \otimes_s dv
      + \bar\gamma_{A'B'} dx^{A'} dx^{B'}  + \frac{1}{M (f+c)^2} df^2, \label{Fmetb}\\
    \gNH(\zeta)  = &
        2 dU dV
        + 2 U d \ln  \left ( \frac{M |F+c]}{F^2} \right ) \otimes_s dV
          + c U^2 \frac{M (F+c)}{F^2}  dV^2 \\
          & + \bar \gamma_{A'B'} dy^{A'} dy^{B'}
        + \frac{dF^2}{M (F+c)^2}. \label{Smetb}
\end{align}
The function $\Xi$ is now (given that $\hat{h}=0$ and $\Delta F
= Q^2(F) P(F) M = (F+c) M$)
\begin{align*}
  \Xi = 1 + UV \frac{(F+c) M}{2F}.
\end{align*}
The functions $C$ and $G$ are, from (\ref{Ch=0}) and (\ref{GXi}),
\begin{align*}
  C = \frac{UV F}{ 1 + UV \frac{(F+c) M}{2F}}
  \left ( 1 - UV \frac{c (F+c) M }{2 F^2} \right ), \quad \quad
  G = \frac{F}{1 + UV \frac{(F+c) M}{2F}}
\end{align*}
and the explicit coordinate change is obtained from (\ref{C}) and (\ref{G})
to be
\begin{align}
  u & = \frac{U}{F} \left ( 1 + UV \frac{(F+c) M}{2F} \right ), \quad \quad
  v = \frac{V F}{ 1 + UV \frac{(F+c) M}{2F}}
  \left ( 1 - UV \frac{c (F+c) M }{2 F^2}  \right ), \nonumber \\
f & =  \frac{1}{\left ( 1 + UV \frac{(F+c) M}{2F} \right )^2}
\left ( F - U V M \frac{c(F+c)(4 F + U V M (F +c))}{4 F^2} \right ),
\quad \quad x^{A'} = y^{A'} \label{isometry2}
  \end{align}
As before an explicit calculation shows that this coordinate change transforms
(\ref{Fmetb}) into (\ref{Smetb}).

\subsection{The case (c)}

Now we have $Q(f) = b + (f+c)^2$ (as before the multiplicative non-zero
constant $Q_0$  can be absorbed in $M$) and
$P(f) = 2 (f+c)/Q$. The functions $h$ and $K$ are, from
(\ref{expshc}) with  $N = Q^2 M$ and (\ref{defK}),
\begin{align*}
  h = M (b + (f+c)^2) , \quad \quad
    K(f) = \frac{( b + c^2 ) (b + (f+c)^2)}{f^2},
\end{align*}
and the two metrics are now
\begin{align}
  \gNH = & 2 du dv + 2 u d \ln  ( M |b + (f+c)^2| ) \otimes_s dv
  + u^2 M (b+ (f+c)^2) dv^2 \nonumber \\
  & + \bar \gamma_{A'B'} dx^{A'} dx^{B'} + \frac{df^2}{M ( b+ (f+c)^2)^2},
  \label{Fmetc} \\
  \gNH (\zeta) = & 2 dU dV + 2 U d \ln \left ( \frac{M | b+ (F+c)^2|}{F^2}
  \right ) \otimes_s dV  + U^2 M \frac{b +c^2}{F^2} \left (
  b + (F+c)^2 \right ) dV^2 \nonumber \\
  & +
  \bar \gamma_{A'B'} dy^{A'} dy^{B'} + \frac{dF^2}{M ( b+ (F+c)^2)^2}.
  \label{Smetc}
\end{align}
The coordinate change is now fairly complicated. The function $\Xi$ is, after
a  calculation that uses
and $N_F = M ( b + (F+c)^2 )^2$ and
$\Delta F = 2 (F+c) (b + (F+c)^2) M$,
\begin{align*}
  \Xi =
  \sqrt{ 1 +  2 U V \frac{ M (F+c)}{F} \left ( b + (F+c)^2 \right )
    - b (UV)^2 \frac{M^2 ( b + (F+c)^2)^2}{F^2}
  }.
\end{align*}
The function $G$ is simply $\displaystyle{\frac{F}{\Xi}}$
 while  $C$ in  (\ref{CXi})
can be rewritten after an algebraic manipulation as
\begin{align*}
  C = \frac{F}{M (b +(F+c)^2)} \left ( \Xi -1 \right )
  -  c UV.
\end{align*}
The explicit form of the coordinate change (\ref{C})-(\ref{G}) turns out to be
\begin{align}
  u &= U \frac{\Xi}{F}, \quad \quad \quad \quad
  v = - c V + \frac{F (\Xi -1)}{U  M (b + (F+c)^2)}, \nonumber \\
  f & = - c + \frac{1}{\Xi}
  \left ( F +c - b UV \frac{M}{F} \left (
  b + (F+c)^2  \right ) \right ), \quad \quad x^{A'} = y^{A'}. \label{isometry3}
\end{align}
Applying this coordinate change to the metric (\ref{Fmetc}) is now fairly
involved but one checks that indeed yields the metric (\ref{Smetc}).

\vspace{3mm}

Summarizing, we have proved the following theorem

\begin{theorem}\label{th:isom}
Let $(M_{\mbox{{\tiny NHG}}},g_{\mbox{{\tiny NHG}}})$ be a near horizon geometry with metric (\ref{NHG}) and assume that the Killing horizon $\Htwo =\{u=0\}$ is multiple with order $m\geq 3$. Let $\zeta\in \Kill^{deg}_{\Htwo}$ be independent of $\eta$. Then, the NHG of $(M_{\mbox{{\tiny NHG}}},g_{\mbox{{\tiny NHG}}})$ with respect to $\zeta$ is locally
isometric to
$(M_{\mbox{{\tiny NHG}}},g_{\mbox{{\tiny NHG}}}(\zeta))$ away from all fixed points
  of $\zeta$  . Moreover,
\begin{itemize}
\item[(i)] If $h=0$ then, in suitable coordinates, the metric $\gNH$ is either
(\ref{Fmeta}) or (\ref{Fmetb}) and the coordinate changes are,
respectively, (\ref{isometry1}) and (\ref{isometry2}).
\item[(ii)] If $h \neq 0$, then $\gNH$ can be written as (\ref{Fmetc})
and the isometry is given by (\ref{isometry3}).
\end{itemize}
\end{theorem}

\subsection{The main theorem}
\begin{theorem}\label{th:main}
Let $(M,g)$ be a spacetime containing a MKH $\H$ with $\dim( \Kill_\H^{deg})\geq 2$
and let $\eta,\zeta\in \Kill_\H^{deg}$.
Then, the NHGs of each connected component
of $\H_{\eta,\zeta}:=\H_{\eta}\cap\H_{\zeta}$ with respect to $\eta$ and $\zeta$ are locally isometric.
\end{theorem}
\begin{proof}
When computing the NHGs of  $(M,g)$ associated to  $\eta$ and $\zeta$, respectively, one starts from  different Gaussian null coordinates.
However, it follows from Remark~\ref{intrinsic} that the NHG is determined by the induced metric $\gamma$, the torsion one-form $\bf{s}$ and the function $h$
on any
cut $S$ of the connected component of $\H_{\eta,\zeta}$. 
For the two Killing vectors these objects are related  as described in (\ref{SH}).
It follows that the NHG  w.r.t.\ $\zeta$  of  $(M,g)$ coincides with the NHG w.r.t.\ $\zeta$  of the NHG  w.r.t.\ $\eta$ of $(M,g)$.

Thus, the NHGs of $(M,g)$ with respect to $\eta$ and $\zeta$ are given by
 $g_{\mbox{{\tiny NHG}}}(\eta)$, say (\ref{NHG}), and $g_{\mbox{{\tiny NHG}}}(\zeta)$, say (\ref{NHGy}), and Theorem \ref{th:isom}  gives the result at once.
\end{proof}

\section*{Acknowledgments}
MM acknowledges financial support under projects
FIS2015-65140-P (Spanish MINECO/FEDER) and
SA083P17 (Junta de Castilla y Le\'on).
TTP acknowledges financial support by the Austrian Science Fund (FWF)
P~28495-N27.
JMMS is supported under Grants No. FIS2017-85076-P (Spanish MINECO/AEI/FEDER, UE) and No. IT956-16 (Basque Government).


\end{document}